\documentclass[11pt,numbers,sort&compress]{elsarticle}
\journal{}

\makeatletter
\def\ps@pprintTitle{%
\let\@oddhead\@empty
\let\@evenhead\@empty
\def\@oddfoot{\hfill\thepage}%
\let\@evenfoot\@oddfoot}
\makeatother

\usepackage{amsmath,amssymb,amsthm,graphicx,hyperref}
\usepackage[labelfont=bf]{caption}
\providecommand{\doi}[1]{\href{https://doi.org/#1}{DOI:#1}}
\usepackage{xurl} % Loads url package with options to break everywhere
\renewcommand{\doi}[1]{%
\href{https://doi.org/#1}{\nolinkurl{DOI:#1}}%
}

\usepackage{natbib} % for \citep and \citet
\usepackage{appendix} % for the appendices
\usepackage[ruled,vlined]{algorithm2e} % for algorithm environments
\usepackage{booktabs} % for \top\mid\bottomrule in tables
\usepackage{multirow} % for \multirow command in tables
\usepackage{float} % more flexible figure placement
\usepackage{enumerate} % For more flexible item labeling
\usepackage{dsfont} % For indicator functions : command \ind
\usepackage{xcolor} % For colored comments \fred and \Guanjie
\usepackage{geometry} % To control the margins
\numberwithin{equation}{section} % equation labeling within each section
\numberwithin{table}{section} % table labeling within each section
\numberwithin{figure}{section} % figure labeling within each section

\theoremstyle{plain}
\newtheorem{theorem}{Theorem}[section]
\newtheorem{proposition}[theorem]{Proposition}
\newtheorem{lemma}[theorem]{Lemma}

\newtheorem{corollary}[theorem]{Corollary}
\newtheorem{assumption}{Assumption}

\newcommand{\EE}{\mathsf{E}}
\newcommand{\PP}{\mathsf{P}}
\newcommand{\Var}{\mathsf{Var}}
\newcommand{\rd}{\mathrm{d}}
\newcommand{\bb}[1]{\boldsymbol{#1}}
\newcommand{\ind}{\mathds{1}}
\newcommand{\OO}{\mathcal{O}}
\newcommand{\oo}{\mathrm{o}}

\allowdisplaybreaks

\makeatletter
\newcommand{\shortname}[1]{\@eadauthor={#1}}
\makeatother

\begin{document}

\begin{frontmatter}
\title{Tweedie-based nonparametric estimation for semicontinuous mixed densities\vspace{-3mm}}

\author[a1]{Guanjie Lyu}\shortname{G.~Lyu}\ead{glyu@dal.ca}
\author[a2]{Fr\'ed\'eric Ouimet}\shortname{F.~Ouimet}\ead{frederic.ouimet2@uqtr.ca}
\author[a1]{Cindy Feng}\shortname{C.~Feng}\ead{cindy.feng@dal.ca}

\address[a1]{Department of Community Health and Epidemiology, Faculty of Medicine, Dalhousie University, Canada}
\address[a2]{D\'epartement de math\'ematiques et d'informatique, Universit\'e du Qu\'ebec \`a Trois-Rivi\`eres, Canada\vspace{-5mm}}

\begin{abstract}
Semicontinuous outcomes occur frequently in health services, insurance, and cost studies. Standard nonparametric density estimators are not well suited to such data because they do not naturally accommodate the mixed structure, the nonnegative support, or the pronounced boundary effects near zero. To address these limitations, we introduce an asymmetric kernel estimator for mixed densities on $[0,\infty)$ based on the Tweedie distribution. For a power parameter $p\in(1,2)$, the Tweedie kernel itself has a point mass at zero and an absolutely continuous component on $(0,\infty)$, yielding a unified smoothing construction that preserves the atom at zero and smooths the positive component using the full semicontinuous sample. We establish pointwise bias and variance expansions, derive asymptotic formulae for the mean squared error and mean integrated squared error, obtain optimal bandwidth rates, and prove asymptotic normality. We propose a profile least-squares cross-validation procedure to jointly select the bandwidth and the power parameter. Simulation results show competitive performance, particularly in challenging boundary-spike and heavy-tailed settings, and an application to emergency department length-of-stay data illustrates the practical value of the method.
\end{abstract}


\end{frontmatter}

\thispagestyle{empty}

\section{Introduction}

Semicontinuous outcomes, characterized by a point mass at zero together with a continuously distributed positive component, arise naturally in many applied settings where no-event or nonuse observations coexist with positive, often strongly right-skewed realizations. Such data have long been recognized in statistics; an early treatment is due to \citet{Aitchison1955semicontinuous}. Since then, a substantial modeling literature has developed, including two-part and marginalized two-part formulations for cross-sectional and longitudinal settings; see, among many others, \citet{Olsen2001TwoPart}, \citet{Smith2014marginalized}, \citet{Smith2017Semicts}, and \citet{Yiu2018two}, as well as the survey by \citet{Min2022semicontinuous}. From a distributional point of view, however, semicontinuous data remain intrinsically challenging because the point mass at zero, the nonnegative support, and the shape of the positive component must all be accommodated simultaneously. Much of this literature has focused on parametric or semiparametric modeling of semicontinuous responses, particularly through two-part formulations, whereas direct nonparametric estimation of the full mixed density has received comparatively less attention. In this setting, flexible nonparametric approaches are valuable because they represent the underlying distribution in a data-driven way, avoid restrictive assumptions, facilitate model assessment, and capture complex features such as boundary behavior and tail structure.

From a nonparametric perspective, this leaves an important gap. Direct estimation of the full mixed density is challenging because the target combines a point mass at zero with a continuously distributed positive component on a nonnegative support. Any nonparametric approach must therefore address two related difficulties simultaneously: the boundary behavior of the estimator near zero and the need to accommodate the discrete and continuous parts within a single distributional target. Classical Rosenblatt-type symmetric kernel estimators \citep{Rosenblatt56} do not respect the support constraint and therefore suffer from boundary bias in neighborhoods of the origin \citep[see, e.g.,][]{Schuster1985,Jones1993,MarronRuppert1994}. This limitation has motivated an extensive literature on boundary correction; see, for example, \citet{Cowling1996Boundary} and the discussion by \citet{KarunamuniAlberts2005}. An alternative strategy is to use asymmetric kernel smoothers whose support matches that of the target density and whose shape varies with the point of estimation. Early positive-support constructions of this type include the smoothed-histogram approach of~\citet{Gawronski1980,GawronskiStadtmuller1981}, while \citet{Chen1999,Chen2000} introduced Beta and Gamma kernel estimators for densities supported on $[0,1]$ and $[0,\infty)$, respectively. Subsequent work has established consistency and refined asymptotic properties for related beta- and gamma-based smoothers, including exact $L_1$ results, weak and uniform consistency, bias correction, multivariate extensions, and density-derivative estimation; see, e.g., \cite{Bouezmarni2003consistency,Bouezmarni2005consistency,FunkeKawka2015MBC,Hirukawa2022UniformBeta,FunkeHirukawa2024Derivative,Funke2025UniformConsistency}. Related asymmetric-kernel methods have also been developed for discontinuity testing and threshold estimation in positively supported skewed densities \citep{FunkeHirukawa2019Discontinuity,FunkeHirukawa2025Splicing}. However, these developments are aimed primarily at fully continuous densities. By contrast, the present problem is genuinely semicontinuous, so that the point mass at zero and the positive component must be handled jointly within a single nonparametric construction.

A natural way to bridge these perspectives is to exploit a distributional family whose intrinsic support and internal structure already mirror those of semicontinuous data. For a power parameter $p\in(1,2)$, Tweedie exponential dispersion models are compound Poisson--Gamma laws on $[0,\infty)$ with a positive point mass at zero and an absolutely continuous component on $(0,\infty)$, as noted by \citet{Tweedie1984Original} and developed systematically by \citet{Jorgensen1987ExponentialDispersion,Jorgensen1997theory}. This observation suggests a kernel construction that differs fundamentally from standard asymmetric kernel smoothers: rather than adapting a purely continuous kernel to the positive half-line, one may use a kernel whose own distributional form already respects both the boundary and the mixed discrete--continuous character of the target. Motivated by this idea, we develop a Tweedie-based nonparametric estimator for the full semicontinuous density, allowing the local amount and shape of smoothing to vary through the Tweedie variance function while retaining the atom at zero within a unified framework.

In this paper, we introduce a Tweedie kernel estimator for the full mixed density $g$ on $[0,\infty)$. Unlike conventional two-part smoothing procedures, which estimate the positive component using only the positive observations, the proposed estimator is built from the entire semicontinuous sample within a single kernel-averaging construction. At $x = 0$, it coincides exactly with the empirical proportion of zeros; for $x > 0$, it yields asymmetric smoothing on the positive half-line, with the zero observations continuing to affect estimation near the boundary. This feature is particularly important near zero, where standard symmetric kernel smoothers suffer from boundary bias and where the coexistence of an atom at zero with small positive observations makes the shape of the continuous component especially difficult to recover. The asymptotic analysis establishes pointwise bias and variance expansions, derives asymptotic formulae for both the mean squared error (MSE) and mean integrated squared error (MISE) along with their corresponding optimal bandwidth rates, and proves asymptotic normality under regularity conditions tailored to the behavior of the positive component near the origin. A profile least-squares cross-validation procedure is also developed for the joint selection of the bandwidth and the Tweedie power parameter, allowing the method to adapt both the amount of smoothing and the kernel shape. Finite-sample performance is then studied under both correctly specified and misspecified semicontinuous models, and an application to emergency department length-of-stay data illustrates the estimator's practical use.

The remainder of the paper is organized as follows. Section~\ref{sec:estimator} introduces the semicontinuous framework and defines the Tweedie kernel estimator. Section~\ref{sec:main.results} develops the asymptotic theory. Section~\ref{sec:simulations} investigates finite-sample performance through simulation studies. Section~\ref{sec:real.data.application} presents the emergency department data application. Section~\ref{sec:conclusion} presents concluding remarks. The appendices supply supplementary technical and empirical details. Specifically, \ref{app:auxiliary.results} collects necessary auxiliary mathematical results, such as near-zero bounds and moment conditions, which serve as the foundation for the detailed proofs of the main theoretical results presented in \ref{app:proofs.main.results}. Furthermore, \ref{app:additional.tables} provides comprehensive numerical tables detailing the Monte Carlo performance metrics. Finally, to ensure full reproducibility, a link to the complete \textsf{R} code used to generate all figures, run the simulation studies, and conduct the real-data application is provided at the end of the manuscript.

\newpage
\section{Semicontinuous data and Tweedie kernel estimator}\label{sec:estimator}

Formally, a random variable $X$ is \emph{semicontinuous} if it admits the representation
\[
X \stackrel{d}{=}
\begin{cases}
0, & \text{with probability } p_0, \\
Y, & \text{with probability } 1 - p_0,
\end{cases}
\]
where $p_0\in(0,1)$ and $Y$ has distribution function $F$ supported on $(0,\infty)$ with density $f$ with respect to the Lebesgue measure. Equivalently, the cumulative distribution function $G$ of $X$ has the form
\begin{equation}\label{eq:semicontinuous}
G(x) = p_0 \ind(x \geq 0) + (1 - p_0) F(x), \qquad x\in \mathbb{R},
\end{equation}
where $\ind(\cdot)$ is the indicator function. The corresponding mixed density can be written as
\[
g(x) = p_0 \ind(x = 0) + (1 - p_0) f(x) \ind(x > 0), \qquad x\in \mathbb{R}.
\]
For $x\in (0,\infty)$, we write $g_{+}(x) = g(x) = (1-p_0) f(x)$.

The presence of an atom at zero violates the assumptions underlying classical kernel density estimators, which are typically constructed for fully continuous densities and rely on local smoothness in a neighborhood of each evaluation point. In particular, standard symmetric kernels produce boundary bias near zero and are ill-suited for distributions combining discrete and continuous components. Consequently, estimation for semicontinuous data requires methods that explicitly account for the mixed nature of the underlying distribution.

A natural nonparametric approach is to focus on estimating the continuous component density $f$ on $(0,\infty)$ while treating the point mass $p_0$ separately. This perspective aligns with two-part modeling frameworks and allows the use of kernel-based techniques adapted to the positive support~\citep{Aitchison1955semicontinuous, Min2022semicontinuous}. In what follows, we consider nonparametric estimation of $g$ through a single kernel-averaging construction and develop a kernel estimator based on the Tweedie family~\citep{Tweedie1984Original} that is adapted to semicontinuous distributions.

Let $X_1,\ldots,X_n$ be independent and identically distributed (i.i.d.) observations from the semicontinuous distribution in~\eqref{eq:semicontinuous}. For each evaluation point $x\geq 0$, let $K_h(\cdot; x)$ denote the Tweedie kernel with mean parameter $\mu = x$, dispersion parameter $h > 0$, and fixed power parameter $p\in(1,2)$. When $x > 0$, this kernel is a compound Poisson--Gamma distribution supported on $[0,\infty)$, with a positive atom at zero and an absolutely continuous component on $(0,\infty)$. Writing
\[
\lambda_x = \frac{x^{2 - p}}{h(2 - p)},\qquad \alpha = \frac{2 - p}{p - 1},\qquad \beta_x = h(p - 1)x^{p - 1},
\]
the kernel is given by
\[
K_h(t; x) = e^{-\lambda_x } \, \ind(t = 0) + e^{-\lambda_x} \sum_{j=1}^{\infty} \frac{\lambda_x ^j}{j!} \, \frac{t^{j\alpha - 1}e^{-t/\beta_x }}{\Gamma(j\alpha)\beta_x ^{j\alpha}} \, \ind(t > 0), \qquad t\geq 0, ~x > 0.
\]
Here $K_h(0; x) = e^{-\lambda_x}$ denotes the point mass assigned to zero, whereas $K_h(t; x)$ for $t > 0$ denotes the Lebesgue subdensity of the absolutely continuous part. When $x = 0$, the Tweedie distribution degenerates at zero; accordingly, the kernel is defined by
\[
K_h(t; 0) = \ind(t = 0),\qquad t\ge0.
\]
The dispersion parameter $h$ plays the role of a bandwidth and controls the degree of local smoothing. The Tweedie kernel estimator of the mixed density $g$ is defined on the full support $[0,\infty)$ by
\[
\widehat{g}_h(x) = \frac{1}{n}\sum_{i=1}^n K_h(X_i; x), \qquad x\geq 0.
\]
Thus, the estimator is obtained by averaging Tweedie kernel contributions over the full semicontinuous sample. For $x > 0$, each observation contributes a location-dependent asymmetric kernel adapted to the positive support and the boundary at zero, while the special definition at $x = 0$ ensures that the atom at the origin is treated exactly within the same construction.

The definition of the kernel at the origin implies that the estimator treats the point mass at zero exactly. Indeed, since $K_h(t; 0) = \ind(t = 0)$, the induced estimator of the zero probability is
\[
\widehat{p}_0 = \widehat{g}_h(0) = \frac{1}{n}\sum_{i=1}^n K_h(X_i; 0) = \frac{1}{n}\sum_{i=1}^n \ind(X_i = 0).
\]
Thus, the implied estimator of the atom at zero coincides exactly with the empirical proportion of zeros.

For $x > 0$, the estimator $\widehat{g}_h(x)$ is still constructed from the full semicontinuous sample. Positive observations contribute through the continuous part of the Tweedie kernel, while zero observations contribute through the atom $K_h(0; x) = e^{-\lambda_x }$. In contrast to conventional two-part procedures, which estimate the positive component using only the positive observations after estimating the zero mass separately, the proposed Tweedie kernel estimator preserves the mixed structure within a single smoothing framework. This distinction is especially relevant near the origin, where the boundary, the atom at zero, and small positive observations interact most strongly. The effect is illustrated below and examined more systematically in the simulation study.

To make this distinction concrete, we generate a small sample of size $n = 15$ from a Tweedie distribution with power parameter $p = 1.5$, mean $\mu = 2$, and dispersion parameter $\phi$ chosen so that the probability of an exact zero is approximately $0.25$. This setting yields a semicontinuous sample containing roughly three or four zero observations, with the remaining observations distributed over the positive half-line. Figure~\ref{fig:illustration} displays the individual kernel contributions (dashed curves) for the two competing approaches, together with their aggregate estimate (dotted curve). In the left panel, the Tweedie kernel estimator assigns a Tweedie kernel to every observation, including the zero-valued observations, shown by the red dashed curves and red rug marks. As a result, the atom at zero contributes directly to estimation of the positive-component density near the boundary. In the right panel, the conventional two-part Gamma kernel estimator discards the zero observations and constructs the density estimate using only the positive observations, represented by the blue dashed curves.

\medskip
\begin{figure}[ht]
\centering
\includegraphics[width=0.49\textwidth, trim=2mm 1mm 2mm 1mm, clip]{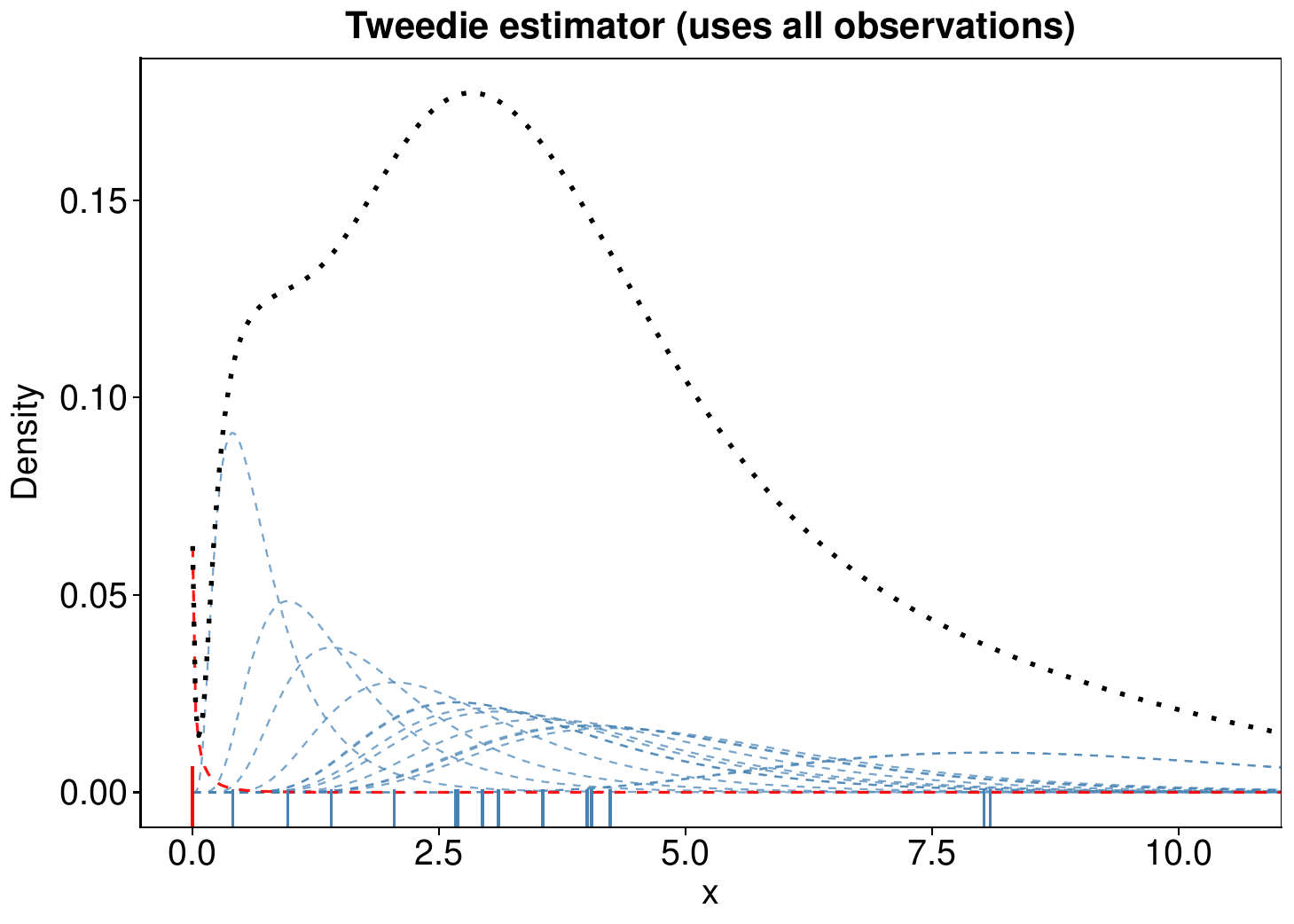}
\hfill
\includegraphics[width=0.49\textwidth, trim=2mm 1mm 2mm 1mm, clip]{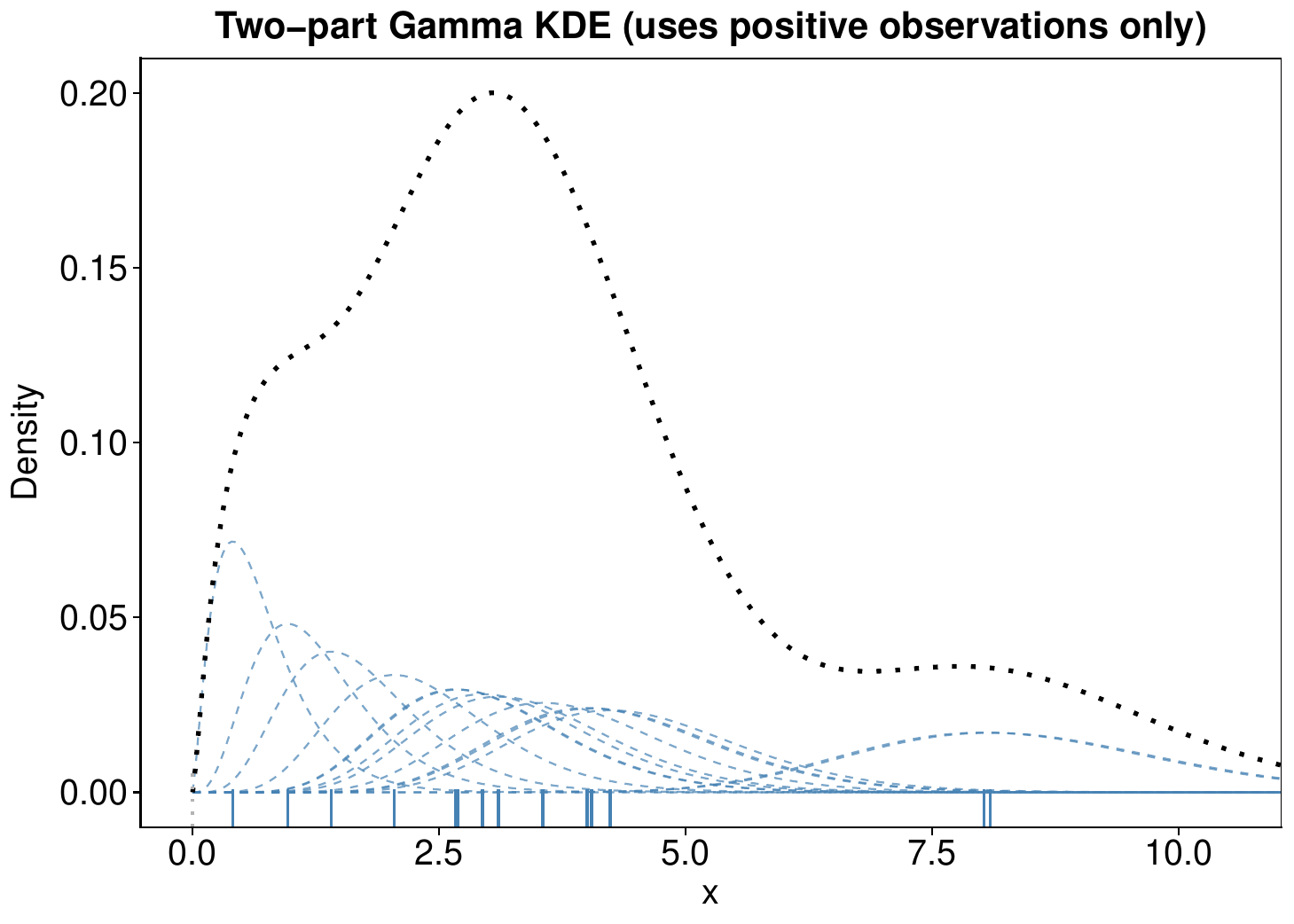}
\caption{Illustration of the estimator construction under a semicontinuous sample with positive observations and an atom at zero. The left panel illustrates the proposed Tweedie kernel estimator as the average of Tweedie kernels centered at each observation, including both zero and positive values, whereas the right panel shows a conventional two-part approach based only on the positive observations. The blue dashed curves correspond to the individual kernel contributions from the positive observations, and the black dotted curve represents the corresponding aggregate estimate. The red marks on the horizontal axis of the left panel indicate the locations of the zero observations.}
\label{fig:illustration}
\end{figure}

\section{Main results}\label{sec:main.results}

Because $\widehat{g}_h(0) = \widehat{p}_0$, the only nontrivial asymptotic analysis concerns the case $x > 0$. The results in this section are therefore pointwise at a fixed interior location $x$, but they must still account for the influence of the boundary at zero through the support of the Tweedie kernel. The assumptions below separate these two aspects: local smoothness around the evaluation point and integrability of the positive component near the origin.

\begin{assumption}\label{ass:smooth}
The density $f$ of the positive component is twice continuously differentiable in a neighborhood of any fixed $x > 0$, and $f''$ is bounded on that neighborhood.
\end{assumption}

\begin{assumption}\label{ass:parameter}
As $n \to \infty$, we have $h = h(n) \to 0$.
\end{assumption}

\begin{assumption}\label{ass:parameter2}
As $n \to \infty$, we have $nh^{1/2} \to \infty$.
\end{assumption}

\begin{assumption}\label{ass:boundary}
Let $x\in(0,\infty)$ be fixed, and let $r = 2 + \delta$ for some $\delta > 0$. Define $\alpha = (2 - p)/(p - 1)$. Assume that there exists $\varepsilon\in(0,x/2)$ such that
\[
\int_0^{\varepsilon} u^{\, r(\alpha - 1)} f(u)\, \rd u < \infty.
\]
\end{assumption}

A convenient sufficient condition for Assumption~\ref{ass:boundary} is that there exist constants $C_f > 0$, $\varepsilon\in(0,x/2)$, and $\kappa > -1$ such that
\[
f(u)\leq C_f u^{\kappa}, \qquad 0 < u\leq \varepsilon,
\]
with $\kappa > r(1 - \alpha) - 1$. Indeed, under this bound,
\[
u^{r(\alpha - 1)}f(u)\leq C_f u^{\kappa + r(\alpha - 1)},
\]
and the right-hand side is integrable near zero precisely when $\kappa + r(\alpha - 1) > -1$. When $p\leq 3/2$, one has $\alpha\ge1$, so $u^{r(\alpha - 1)}$ is bounded near zero and therefore Assumption~\ref{ass:boundary} is automatically satisfied for any positive-component density $f$.

The roles of the assumptions are as follows. Assumption~\ref{ass:smooth} provides the local smoothness needed for the Taylor expansion at the evaluation point, Assumption~\ref{ass:parameter} places the estimator in the vanishing-bandwidth regime, and Assumption~\ref{ass:parameter2} enters the asymptotic normality result through the effective sample size $nh^{1/2}$. Assumption~\ref{ass:boundary} is specific to the mixed Tweedie construction and controls the contribution from the near-zero region. Although the asymptotic analysis is carried out at a fixed interior point $x > 0$, the Tweedie kernel has support on $[0,\infty)$, and observations close to the boundary can contribute through the left tail of the kernel.

This boundary contribution is the main additional technical feature relative to more classical positive-support asymmetric-kernel estimators. Near zero, the absolutely continuous component of the Tweedie kernel behaves like $u^{\alpha - 1}$, where $\alpha = (2 - p)/(p - 1)$. Consequently, the contribution from a neighborhood of the origin depends on the integrability of $u^{r(\alpha - 1)}f(u)$, which is precisely the role of Assumption~\ref{ass:boundary}. The requirement becomes more restrictive as $p\uparrow 2$: since $\alpha\downarrow 0$, the sufficient condition $\kappa > r(1 - \alpha) - 1$ approaches $\kappa > r - 1 = 1 + \delta$, whereas for $p\le3/2$ one has $\alpha\ge1$, so boundedness of $f$ near zero is enough. Proposition~\ref{prop:near-zero} shows that, under this condition, the near-zero contribution is exponentially small in $h^{-1}$. This allows the fixed-$x$ bias and variance expansions below to retain the same leading-order form as in the gamma-kernel case, despite the additional boundary difficulty introduced by the mixed Tweedie kernel.

The first result gives the pointwise bias expansion. Its leading term is of order $h$, with coefficient determined by the local curvature of $g$ and the Tweedie variance structure through the factor $x^p$.

\begin{theorem}\label{thm:bias_tkde}
Suppose that Assumptions~\ref{ass:smooth},~\ref{ass:parameter} and~\ref{ass:boundary} hold. Then, for any fixed $x\in (0,\infty)$, as $n \to \infty$,
\[
\EE\{\widehat{g}_h(x)\} - g(x) = \frac{1}{2} \, h \, x^{p} g''(x) + \oo(h).
\]
\end{theorem}

The next theorem provides the matching pointwise variance expansion. Its leading order, $n^{-1}h^{-1/2}$, determines the stochastic scale of the estimator and will later induce the normalization used in the asymptotic normality result.

\begin{theorem}\label{thm:var}
Suppose that Assumptions~\ref{ass:smooth},~\ref{ass:parameter} and~\ref{ass:boundary} hold. Then, for any fixed $x\in (0,\infty)$, as $n \to \infty$,
\[
\Var\bigl\{\widehat{g}_h(x)\bigr\} = \frac{g(x)}{2\sqrt{\pi} \, x^{p/2}} n^{-1}h^{-1/2} + \oo({n}^{-1}h^{-1/2}).
\]
\end{theorem}

Combining the preceding bias and variance expansions yields the usual asymptotic risk approximations. The next two corollaries provide the pointwise $\mathrm{MSE}$ and $\mathrm{MISE}$ expansions, together with the corresponding optimal bandwidth rates.

\begin{corollary}\label{cor:mse-opt}
Suppose that Assumptions~\ref{ass:smooth},~\ref{ass:parameter} and~\ref{ass:boundary} hold. Then, for any fixed $x\in (0,\infty)$, as $n \to \infty$,
\[
\mathrm{MSE}\{\widehat{g}_h(x)\} = \frac{1}{4} x^{2p} \{g''(x)\}^2 h^2 + \frac{g(x)}{2\sqrt{\pi} \, x^{p/2}} n^{-1} h^{-1/2} + \oo\bigl(h^2 + n^{-1} h^{-1/2}\bigr).
\]
The asymptotically MSE-optimal bandwidth is
\[
h_{\mathrm{opt}}(x) = \left(\frac{g(x)}{2\sqrt{\pi} \, x^{5p/2} \{g''(x)\}^2}\right)^{2/5} n^{-2/5},
\]
and the corresponding optimal MSE satisfies
\[
\mathrm{MSE}\{\widehat{g}_{h_{\mathrm{opt}}(x)}(x)\} = \frac{5}{4} \left(\frac{g(x)^4 \{g''(x)\}^2}{16\pi^2}\right)^{1/5} n^{-4/5} + \oo(n^{-4/5}).
\]
\end{corollary}

\begin{corollary}\label{cor:imse}
Suppose that Assumptions~\ref{ass:smooth},~\ref{ass:parameter} and~\ref{ass:boundary} hold. In addition, assume that $g$ satisfies
\[
R_p(g'') = \int_0^{\infty} x^{2p} \{g''(x)\}^2 \, \rd x < \infty \qquad \text{and} \qquad S_p(g) = \int_0^{\infty} x^{-p/2} g(x) \, \rd x < \infty.
\]
Then, as $n \to \infty$,
\[
\mathrm{MISE}\{\widehat{g}_h\} = \frac{1}{4} h^2 R_p(g'') + \frac{1}{2\sqrt{\pi}} \, n^{-1} h^{-1/2} S_p(g) + \oo(h^2 + n^{-1} h^{-1/2}).
\]
The asymptotically MISE-optimal bandwidth is
\[
h_{\mathrm{opt}}^{\mathrm{MISE}} = \left(\frac{S_p(g)}{2\sqrt{\pi} \, R_p(g'')}\right)^{2/5} n^{-2/5},
\]
and the corresponding optimal MISE satisfies
\[
\mathrm{MISE}\{\widehat{g}_{h_{\mathrm{opt}}^{\mathrm{MISE}}}\} = \frac{5}{4} \left(\frac{R_p(g'') \, S_p(g)^4}{16 \pi^2}\right)^{1/5} n^{-4/5} + \oo(n^{-4/5}).
\]
\end{corollary}

The final result of this section establishes asymptotic normality. The normalization $n^{1/2} h^{1/4}$ is dictated by the variance rate in Theorem~\ref{thm:var}, while the additional condition $nh^{5/2}\to \lambda$ determines whether the leading bias term remains visible in the limiting distribution after centering at $g(x)$.

\begin{theorem}\label{thm:asymp-normal}
Suppose that Assumptions~\ref{ass:smooth}-\ref{ass:boundary} hold. Then, for any fixed $x\in(0,\infty)$,
\[
n^{1/2} h^{1/4} \big(\widehat{g}_h(x) - \EE\{\widehat{g}_h(x)\}\big) \stackrel{d}{\longrightarrow} \mathcal{N}\left(0,\frac{g(x)}{2\sqrt{\pi} \, x^{p/2}}\right).
\]
If, in addition, $n h^{5/2}\to \lambda\in[0,\infty)$, then
\[
n^{1/2} h^{1/4} \big(\widehat{g}_h(x) - g(x)\big) \stackrel{d}{\longrightarrow} \mathcal{N}\left(\frac{1}{2} x^p g''(x)\sqrt{\lambda}, \frac{g(x)}{2\sqrt{\pi} \, x^{p/2}}\right).
\]
\end{theorem}

\section{Simulation study}\label{sec:simulations}

This section investigates the finite-sample performance of the proposed Tweedie kernel estimator for semicontinuous data. The simulation study is designed to assess how well the estimator recovers the positive component of the target density across a range of settings, including the Tweedie case and several misspecified zero-inflated alternatives with unimodal and bimodal shapes. We compare the proposed method with Gaussian kernel density estimators using plug-in and least-squares cross-validation bandwidth selection, implemented via the \texttt{kde} function in the \textbf{ks} package~\citep{R-ks}, as well as with the Gamma kernel estimator, implemented via the \texttt{kdensity} function in the \textbf{kdensity} package~\citep{R-kdensity}. Performance is evaluated on the positive half-line using integrated squared and absolute error criteria, under varying sample sizes and zero-mass levels, and we also examine the data-driven selection of the smoothing parameter $h$ and the Tweedie power index $p$.

\subsection{Models and performance criteria}

We consider four data-generating processes (DGPs) to evaluate both the performance of the proposed estimator under correct specification and its robustness to model misspecification. Scenario M.1 corresponds to correct specification for the Tweedie kernel estimator, whereas scenarios M.2--M.4 represent increasingly challenging misspecification settings with non-Tweedie positive components, including both unimodal and bimodal shapes. In all cases, we set the sample size $n\in \{100, 200, 500\}$ and the structural-zero probability $p_0\in \{0.15, 0.30, 0.45\}$.

\begin{enumerate}[M.1.]
\item \textit{Tweedie DGP}. We generate $X\sim \mathrm{Tw}_p(\mu,\phi)$ with power index $p = 1.1$ and mean $\mu = 2$. Then $\phi = \mu^{2 - p}/\{(2 - p)(-\log p_0)\}$. For this family, $\PP(X = 0) = p_0(\mu,\phi,p) > 0$. The target function on $(0,\infty)$ is
\[
g_{+}(x) = f_{\mathrm{Tw}}(x; \mu,\phi,p), \qquad x > 0,
\]
where $f_{\mathrm{Tw}}(\cdot; \mu,\phi,p)$ denotes the Tweedie density on $(0,\infty)$.

\item \textit{Zero-inflated Gamma.} We generate
\[
X =
\begin{cases}
0, & \text{with probability } p_0, \\
X_{+}, & \text{with probability } 1 - p_0,
\end{cases}
\qquad X_{+}\sim \mathrm{Gamma}(a,b),
\]
where $a = 1.3$ is the shape parameter and $b = 6$ is the rate parameter. Then the target on $(0,\infty)$ is
\[
g_{+}(x) = (1 - p_0) \, f_{\Gamma}(x; a,b), \qquad x > 0,
\]
with $f_{\Gamma}(\cdot; a,b)$ the Gamma density.

\item \textit{Zero-inflated two-component Gamma mixture (boundary spike and heavy right tail).} We generate
\[
X =
\begin{cases}
0, & \text{with probability } p_0, \\
X_{+}, & \text{with probability } 1 - p_0,
\end{cases}
\qquad X_{+}\sim 0.55 \, \mathrm{Gamma}(2,6) + 0.45 \, \mathrm{Gamma}(15,1),
\]
so that the positive component is sharply peaked near the boundary and has a heavy right tail. The target function on $(0,\infty)$ is
\[
g_{+}(x) = (1 - p_0)\left\{0.55 f_{\Gamma}(x; 2,6) + 0.45 f_{\Gamma}(x; 15,1)\right\}, \qquad x > 0.
\]
\item \textit{Zero-inflated two-component Gamma mixture (bimodal with separate modes).} This scenario is identical to Scenario~M.3, except that the mixture components are selected to yield a bimodal positive component. Specifically,
\[
X_{+}\sim 0.35 \, \mathrm{Gamma}(4,6) + 0.65 \, \mathrm{Gamma}(20,3).
\]
The target function on $(0,\infty)$ is
\[
g_{+}(x) = (1 - p_0)\left\{0.35 f_{\Gamma}(x; 4,6) + 0.65 f_{\Gamma}(x; 20,3)\right\}, \qquad x > 0.
\]
\end{enumerate}
Figure~\ref{fig:true_densities} displays the true positive-component target functions $g_{+}(x)$ under the four simulation scenarios, with the structural-zero probability fixed at $p_0 = 0.3$. Scenario M.1 corresponds to the Tweedie case and exhibits a unimodal shape with moderate right skewness. Scenario M.2 is a zero-inflated Gamma model that is likewise unimodal, but more sharply concentrated near the origin and supported over a shorter range. Scenario M.3 represents a substantially more challenging misspecification setting, with a pronounced boundary spike near zero together with a long and relatively diffuse right tail. Finally, Scenario M.4 is a zero-inflated Gamma mixture with two clearly separated modes, providing a bimodal alternative that departs qualitatively from the Tweedie shape.

All methods are based on the same semicontinuous sample $\{X_i\}_{i=1}^n$. Since the estimated point mass at zero is identical across the competing methods, namely
\[
\widehat{p}_0 = \frac{1}{n}\sum_{i=1}^n \ind(X_i = 0),
\]
comparisons may be restricted to the positive component. Accordingly, for an estimator $\widehat{g}_{+; h}$ with smoothing parameter $h$, we evaluate performance mainly using the integrated squared error over $(0,\infty)$,
\[
\mathrm{ISE}_{+}\left(\widehat{g}_{+; h}\right) = \int_{0}^{\infty}\left\{\widehat{g}_{+; h}(x) - g_{+}(x)\right\}^2 \, \rd x.
\]
We also consider the integrated absolute error,
\[
\mathrm{IAE}_{+}\left(\widehat{g}_{+; h}\right) = \int_{0}^{\infty}\left|\widehat{g}_{+; h}(x) - g_{+}(x)\right| \, \rd x.
\]

\begin{figure}[ht]
\centering
\includegraphics[trim={2mm 3mm 2mm 1mm}, clip, width=0.80\textwidth]{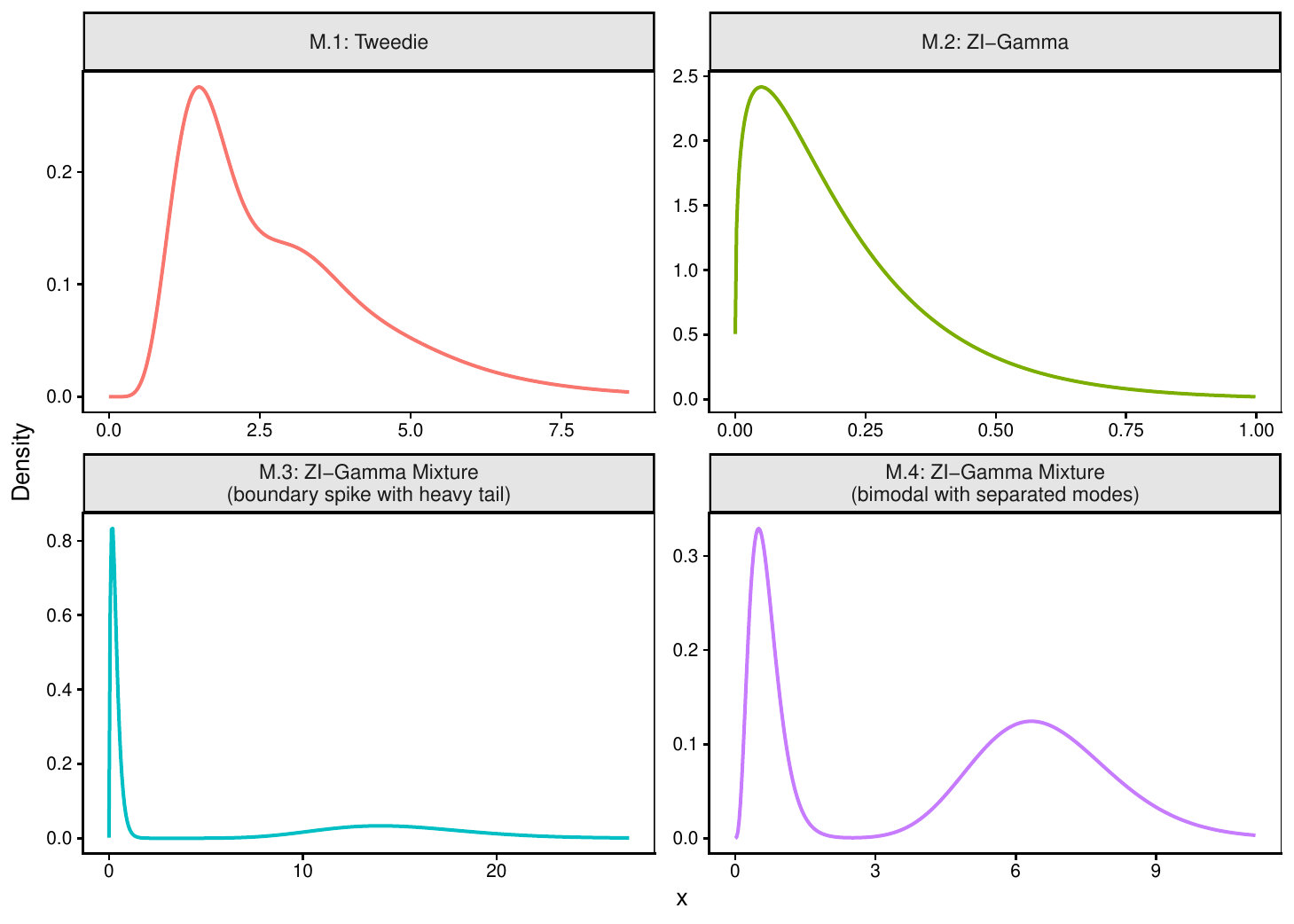}
\caption{True positive-component target functions $g_{+}(x)$ for the four simulation scenarios with $p_0 = 0.3$. M.1: Tweedie distribution. M.2: Zero-inflated Gamma distribution. M.3: Zero-inflated Gamma mixture exhibiting a sharp boundary spike near zero and a heavy right tail. M.4: Zero-inflated Gamma mixture with two well-separated modes.}
\label{fig:true_densities}
\end{figure}

\subsection{Bandwidth selection for the Tweedie kernel estimator via least-squares cross-validation}\label{subsec:tkde-lscv}

For smoothing-parameter selection, given a fixed power index $p$, we choose $h$ by minimizing a least-squares cross-validation (LSCV) criterion designed to target the integrated squared error on the positive half-line. To match the simulation criterion
\[
\mathrm{ISE}_{+}\left(\widehat{g}_{+; h}\right) = \int_{0}^{\infty}\left\{\widehat{g}_{+; h}(x) - g_{+}(x)\right\}^2 \, \rd x,
\]
we restrict the LSCV objective to $(0,\infty)$. Specifically, for an equally spaced grid $\{x_{\ell}\}_{\ell=1}^m \subseteq (0,\infty)$ with common spacing $\Delta x = x_{\ell + 1} - x_{\ell}$, we consider
\[
\mathrm{LSCV}_{+}(h) = \underbrace{\int_{0}^{\infty}\widehat{g}_{+; h}(x)^2 \, \rd x}_{\text{Term 1}} - \underbrace{\frac{2}{n}\sum_{i=1}^n \widehat{g}_{+; h}^{(-i)}(X_i) \, \ind(X_i > 0)}_{\text{Term 2}},
\]
which is approximated numerically by
\[
\widehat{\mathrm{LSCV}}_{+}(h) = \sum_{\ell=1}^m \widehat{g}_{+; h}(x_{\ell})^2 \, \Delta x - \frac{2}{n}\sum_{i:X_i > 0}\widehat{g}_{+; h}^{(-i)}(X_i).
\]
In our implementation, Term~1 is computed by evaluating $\widehat{g}_{+; h}(x_{\ell})$ on the prescribed grid and approximating the integral by a Riemann sum. For Term~2, we avoid recomputing the Tweedie kernel estimator for each leave-one-out sample. Indeed, since
\[
\widehat{g}_{+; h}(x) = \frac{1}{n}\sum_{j=1}^n K_h(X_j; x), \qquad x > 0,
\]
the corresponding leave-one-out estimator admits the representation
\begin{equation}\label{eq:LOO.identity}
\widehat{g}_{+; h}^{(-i)}(x) = \frac{1}{n - 1}\sum_{j\ne i} K_h(X_j; x) = \frac{n}{n - 1}\widehat{g}_{+; h}(x) - \frac{1}{n - 1}K_h(X_i; x).
\end{equation}
This identity permits direct evaluation of $\widehat{g}_{+; h}^{(-i)}(X_i)$ for observations with $X_i > 0$, thereby avoiding repeated refitting and substantially reducing the computational burden of the LSCV procedure.

For each dataset, and for a fixed power index $p$, the bandwidth $h$ is selected by minimizing $\widehat{\mathrm{LSCV}}_{+}(h; p)$ over a prespecified grid $h \in [h_{\min},h_{\max}]$. The resulting minimizer, denoted by $h^{\star}(p)$, is then used to construct the Tweedie kernel estimator corresponding to that value of $p$. To allow for data-driven selection of the power index, we next repeat this bandwidth-selection step over a prespecified grid of candidate values $p \in [p_{\min},p_{\max}]$. Specifically, for each candidate $p$, we compute $h^{\star}(p)$ as described above and then evaluate the corresponding cross-validation criterion. The final power--bandwidth pair is chosen as
\[
(p^{\star},h^{\star}) = \arg\min_{p}\ \widehat{\mathrm{LSCV}}_{+}\bigl(h^{\star}(p); p\bigr),
\]
where $h^{\star} = h^{\star}(p^{\star})$. The selected pair is then used to construct the final Tweedie kernel estimator and to compute the corresponding value of $\mathrm{ISE}_{+}$. The corresponding selection procedure is summarized in Algorithm~\ref{alg:tkde-ph-lscv}. Throughout the simulation study, the bandwidth and power parameters are selected over prespecified grids with $N_h = 20$ candidate values for $h$ and $N_p = 18$ candidate values for $p$.

\begin{algorithm}[ht]
\caption{Profile LSCV selection of the power--bandwidth pair $(p,h)$ for the Tweedie kernel estimator}
\label{alg:tkde-ph-lscv}
\DontPrintSemicolon
\KwIn{Semicontinuous sample $X_1,\ldots,X_n$; evaluation grid $\{x_{\ell}\}_{\ell=1}^m \subseteq (0,\infty)$ with spacing $\Delta x$; candidate power grid $\mathcal{P} = \{p_1,\ldots,p_{N_p}\} \subseteq [p_{\min},p_{\max}]$; candidate bandwidth grid $\mathcal{H} = \{h_1,\ldots,h_{N_h}\}\subseteq [h_{\min},h_{\max}]$.}
\KwOut{Selected pair $(p^{\star},h^{\star})$.}

\For{$k = 1,\ldots,N_p$}{
Set $p \leftarrow p_k$.

\For{$j = 1,\ldots,N_h$}{
Set $h \leftarrow h_j$.

Compute the Tweedie kernel estimator $\widehat{g}_h(\cdot; p)$.

Evaluate
\[
\widehat{\mathrm{LSCV}}_{+}(h; p) = \sum_{\ell=1}^m \widehat{g}_h(x_{\ell}; p)^2 \, \Delta x - \frac{2}{n}\sum_{i:X_i > 0}\widehat{g}_h^{(-i)}(X_i; p),
\]
where $\widehat{g}_h^{(-i)}(X_i; p)$ is computed using the leave-one-out identity \eqref{eq:LOO.identity}.
}

Determine
\[
h^{\star}(p) = \arg\min_{h\in \mathcal{H}}\widehat{\mathrm{LSCV}}_{+}(h; p)
\]
and record
\[
\mathrm{CV}^{\star}(p) = \widehat{\mathrm{LSCV}}_{+}\bigl(h^{\star}(p); p\bigr).
\]
}

Select
\[
p^{\star} = \arg\min_{p\in \mathcal{P}}\mathrm{CV}^{\star}(p), \qquad h^{\star} = h^{\star}(p^{\star}).
\]

\Return $(p^{\star},h^{\star})$
\end{algorithm}

\subsection{Results and discussion}

Figures~\ref{fig:simulationISE} and~\ref{fig:simulationIAE} present a coherent picture of the comparative performance of the four estimators under $\mathrm{ISE}_{+}$ and $\mathrm{IAE}_{+}$. In both figures, estimation accuracy improves with sample size, and the proposed Tweedie kernel estimator is consistently competitive across all scenarios. Under the Tweedie design (M.1) and the unimodal zero-inflated Gamma design (M.2), it performs at least as well as the competing Gaussian and Gamma kernel estimators in most settings. The largest advantage appears in Scenario~M.3, where the positive component exhibits a pronounced boundary spike and a heavy right tail. In this setting, the proposed estimator uniformly attains markedly smaller errors and lower variability under both criteria. This pattern is consistent with the discussion in Section~\ref{sec:estimator}: because the Tweedie kernel estimator is constructed from the full semicontinuous sample, including the observations at zero, it is better able to capture the distributional structure in a neighborhood of the origin, where the interaction between the atom at zero and the continuous positive component is strongest. In the bimodal setting M.4, it again performs very favorably and is typically the best or nearly the best method. The favorable performance in Scenarios~M.3 and~M.4 may also be attributed to the additional flexibility of the proposed procedure, which jointly selects the Tweedie power parameter $p$ and the bandwidth $h$. In contrast to conventional kernel estimators that tune only a smoothing parameter, the proposed estimator adapts both the degree of smoothing and the kernel shape. This added flexibility is especially useful in complex misspecification settings, where the positive component departs substantially from standard unimodal forms.

\begin{figure}[ht]
\centering
\includegraphics[trim = 2mm 1mm 2mm 1mm, clip, width = 0.49\linewidth, height = 0.30\linewidth]{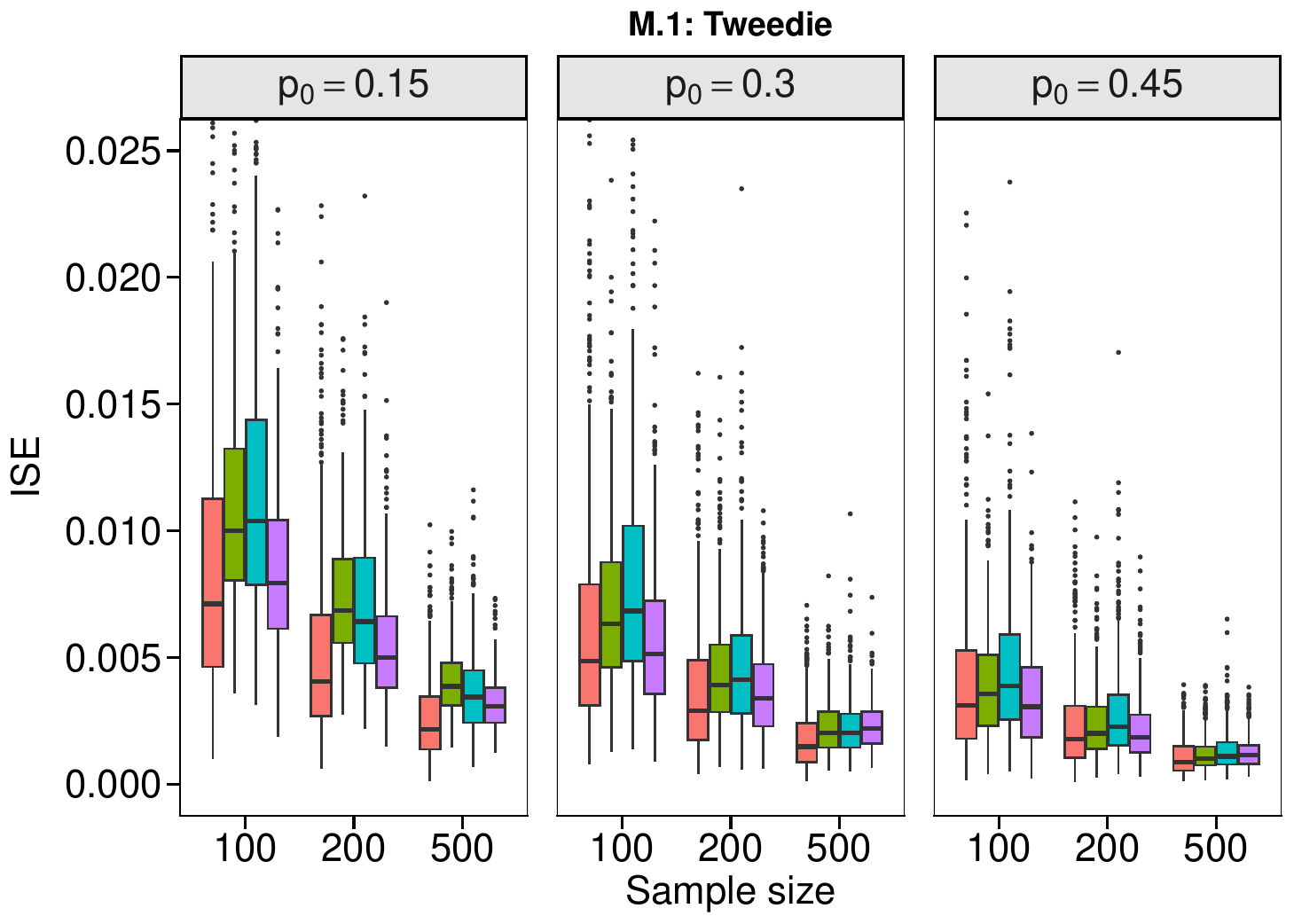}
\hfill
\includegraphics[trim = 2mm 1mm 2mm 1mm, clip, width = 0.49\linewidth, height = 0.30\linewidth]{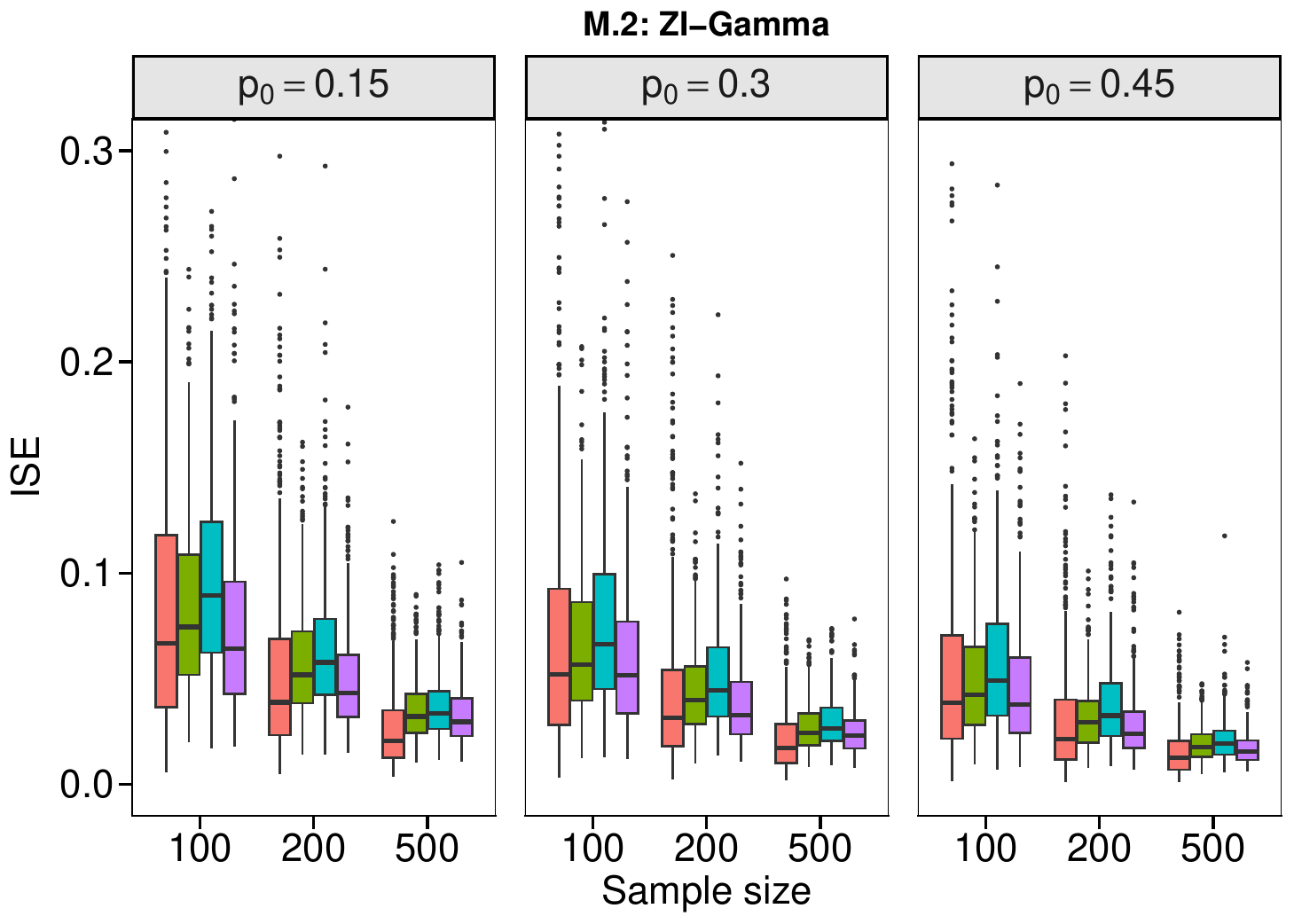}
\includegraphics[trim = 2mm 5mm 2mm 1mm, clip, width = 0.49\linewidth, height = 0.336\linewidth]{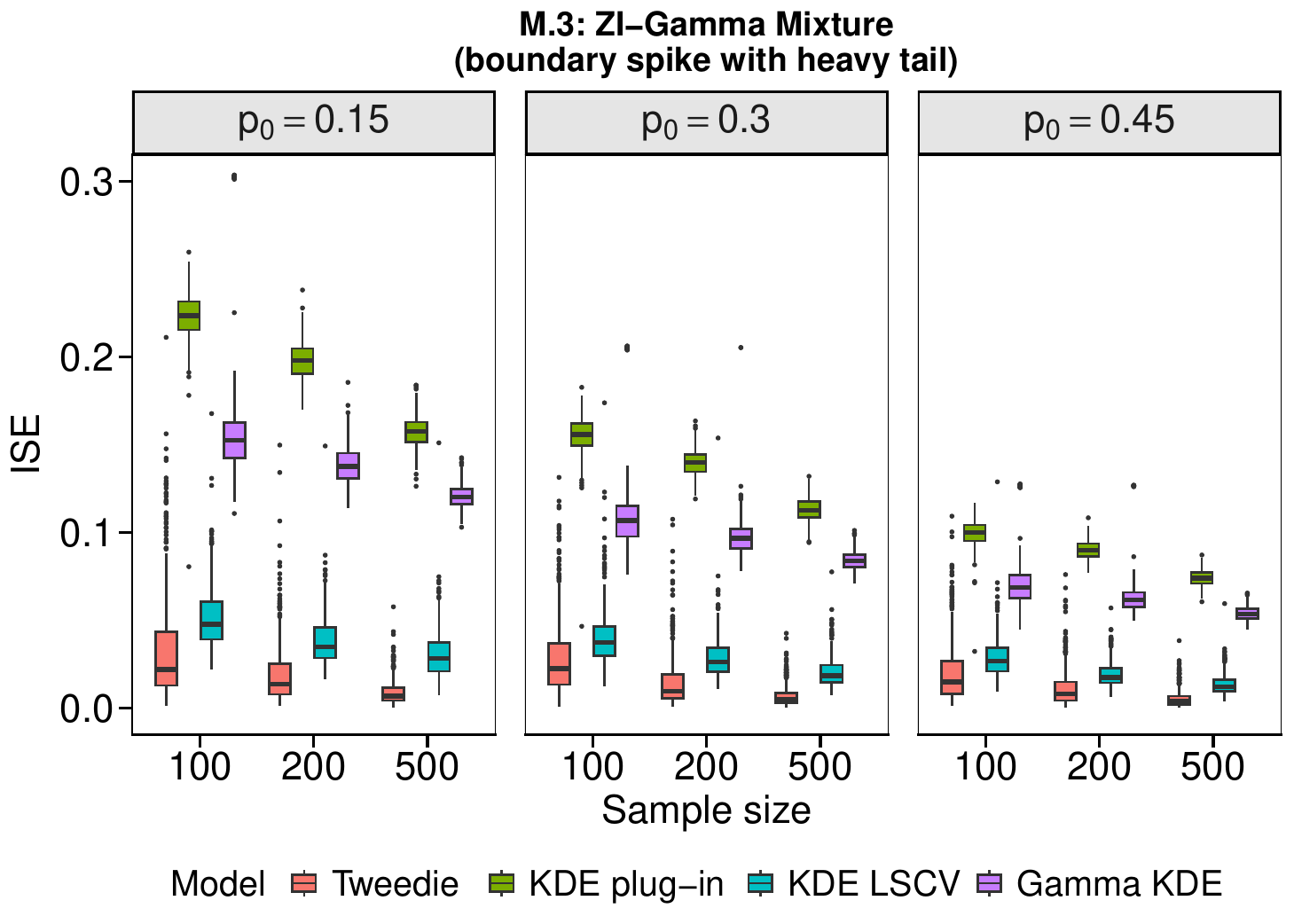}
\hfill
\includegraphics[trim = 2mm 5mm 2mm 1mm, clip, width = 0.49\linewidth, height = 0.336\linewidth]{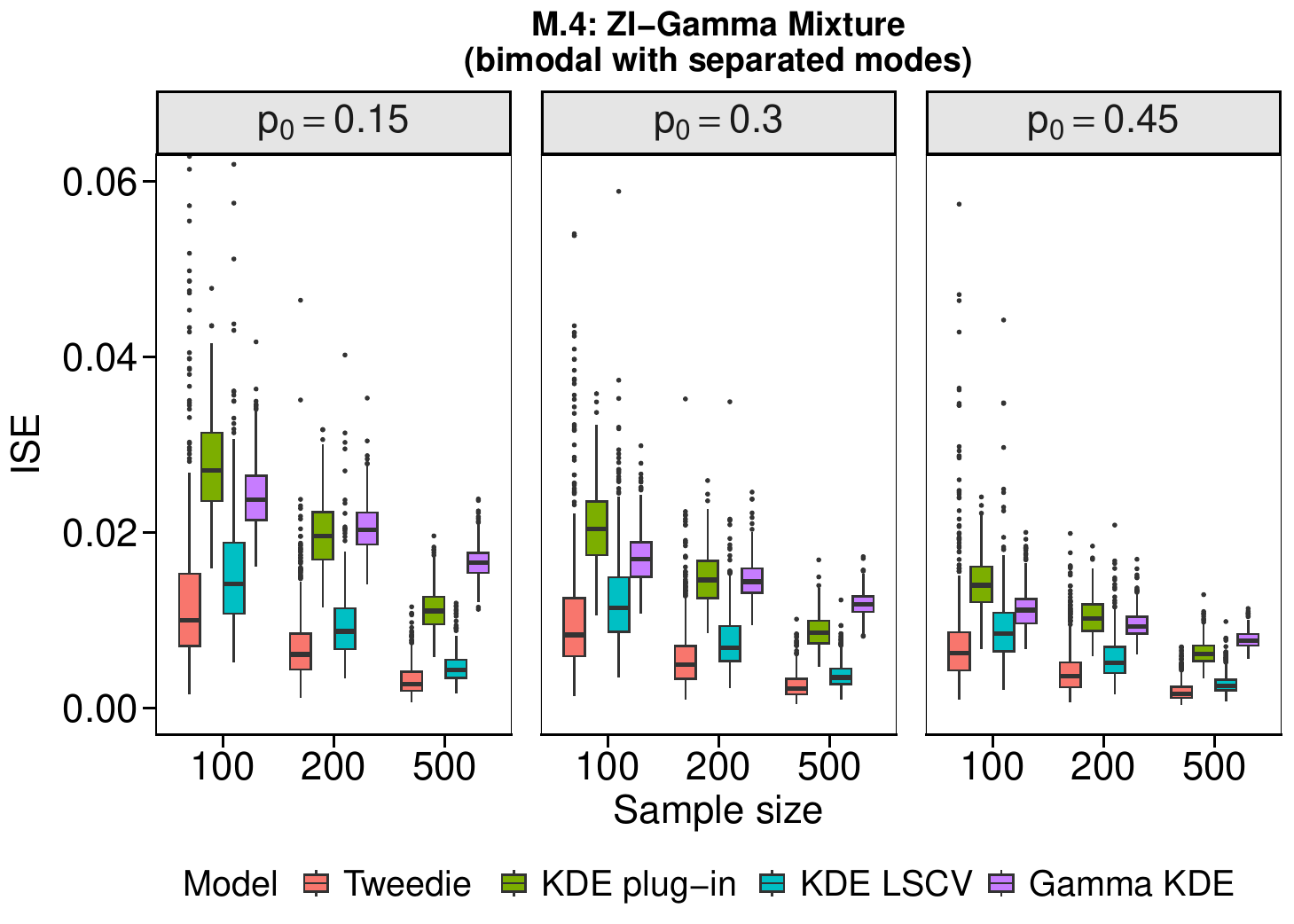}
\caption{Monte Carlo boxplots (500 replicates) of the integrated squared error ($\mathrm{ISE}_{+}$) of the density estimators on the positive half-line $(0,\infty)$ under the four simulation scenarios. Each panel corresponds to a data-generating mechanism (M.1--M.4), with columns showing different values of the zero mass $p_0 \in \{0.15,0.30,0.45\}$ and sample sizes $n \in \{100,200,500\}$. The proposed Tweedie kernel estimator is compared with the Gaussian KDE using plug-in bandwidth selection, the Gaussian KDE with least-squares cross-validation (LSCV), and the Gamma KDE.}
\label{fig:simulationISE}
\end{figure}

\begin{figure}[ht]
\centering
\includegraphics[trim = 2mm 1mm 2mm 1mm, clip, width = 0.49\linewidth, height = 0.30\linewidth]{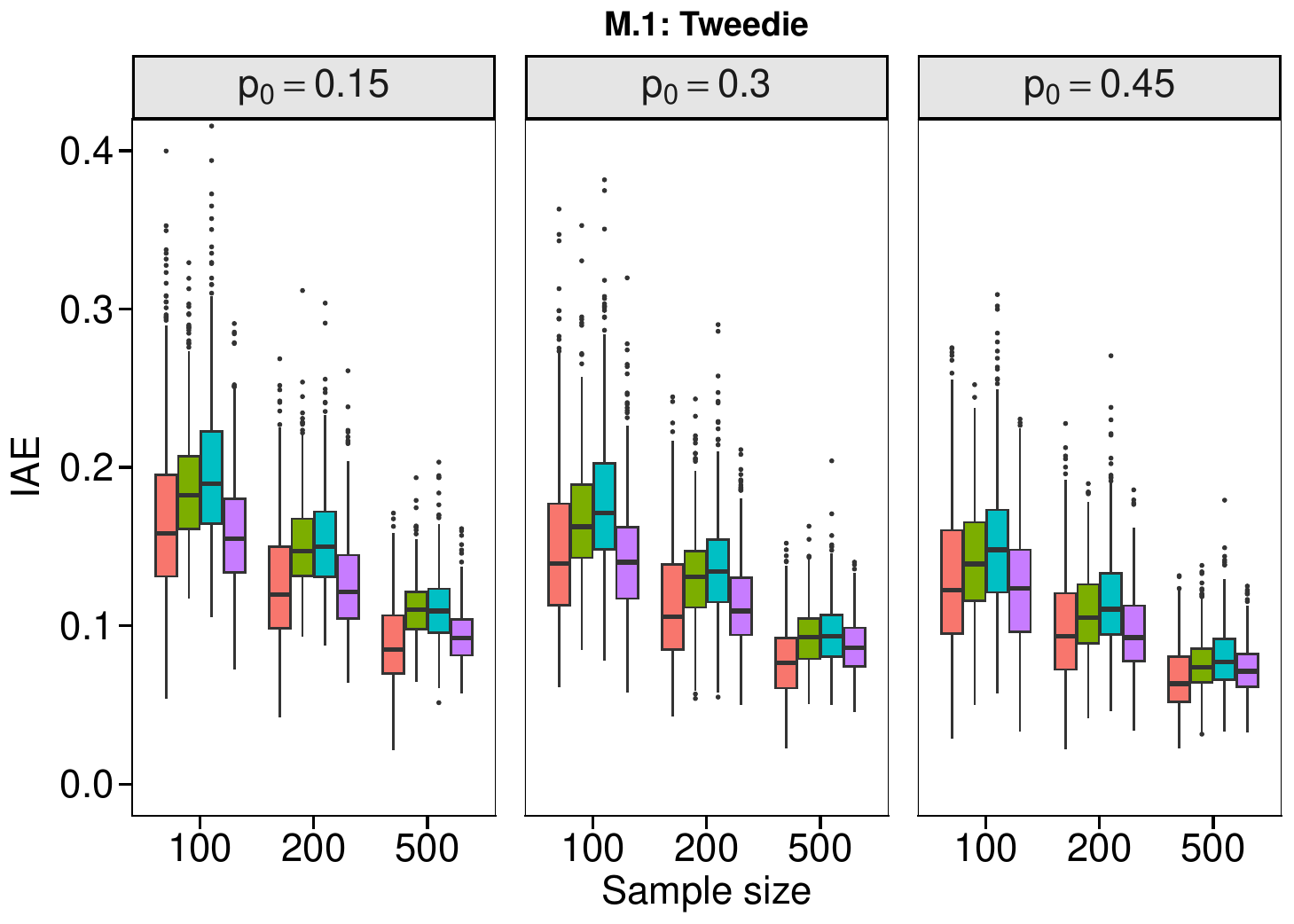}
\hfill
\includegraphics[trim = 2mm 1mm 2mm 1mm, clip, width = 0.49\linewidth, height = 0.30\linewidth]{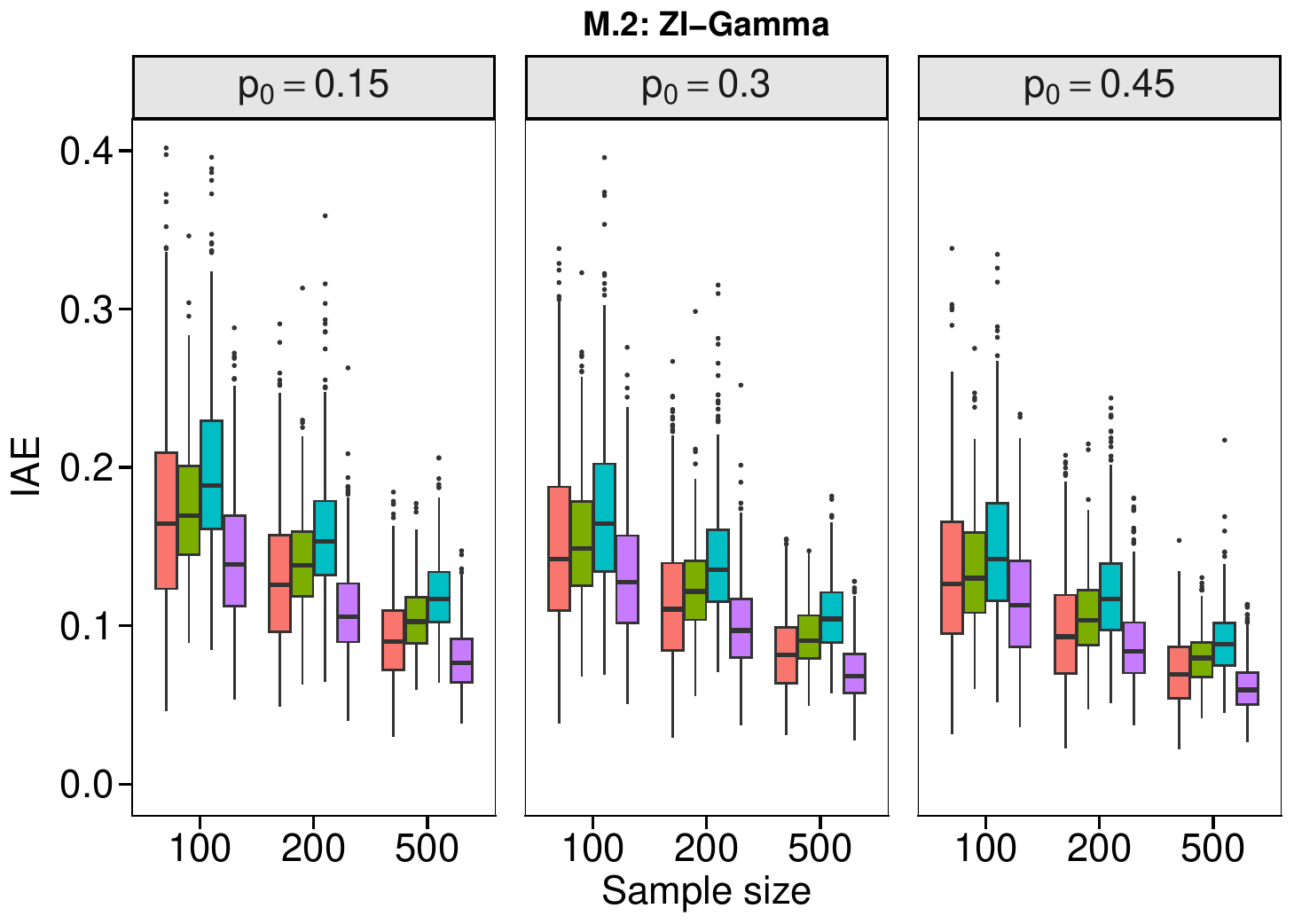} \\
\includegraphics[trim = 2mm 5mm 2mm 1mm, clip, width = 0.49\linewidth, height = 0.336\linewidth]{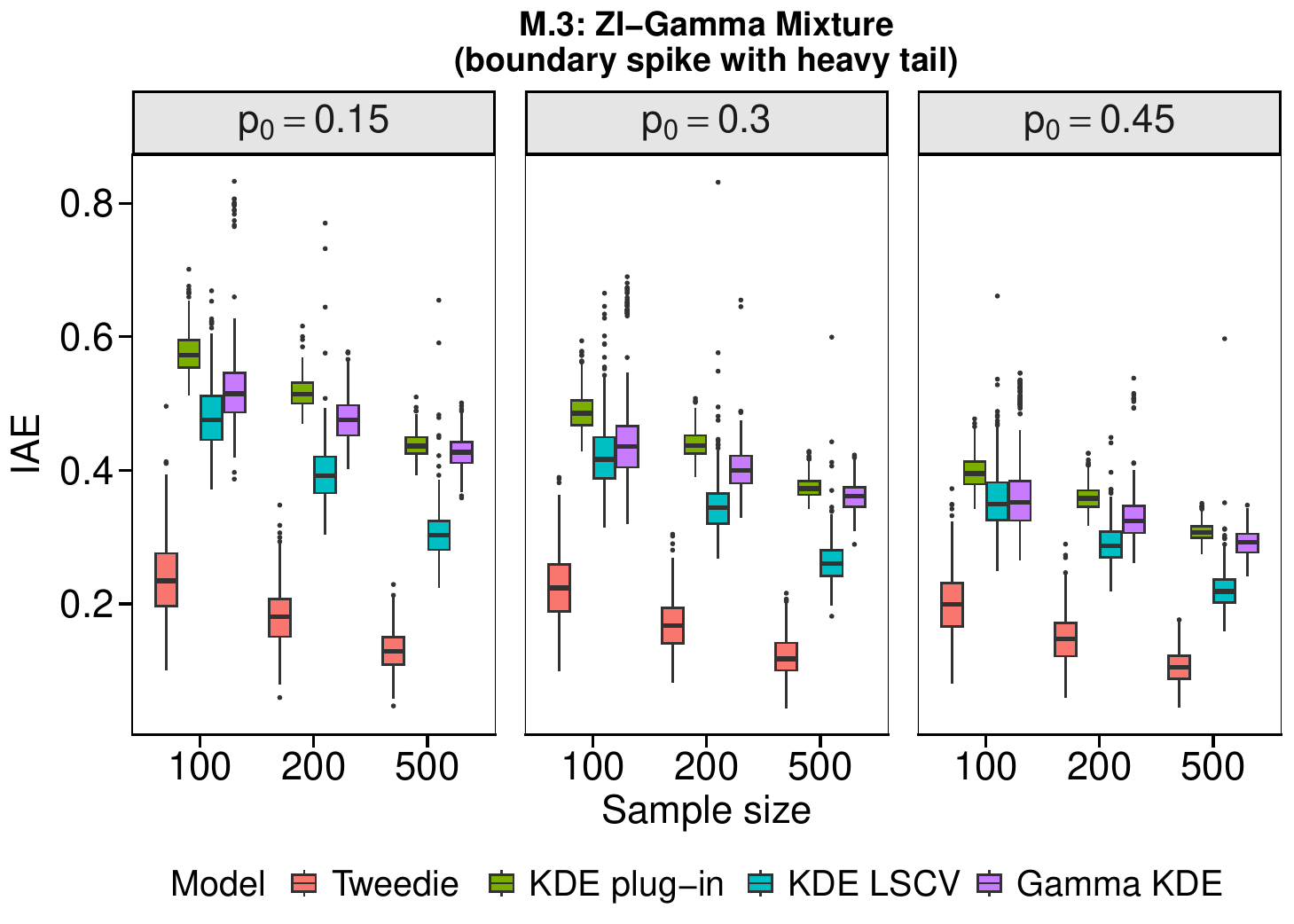}
\hfill
\includegraphics[trim = 2mm 5mm 2mm 1mm, clip, width = 0.49\linewidth, height = 0.336\linewidth]{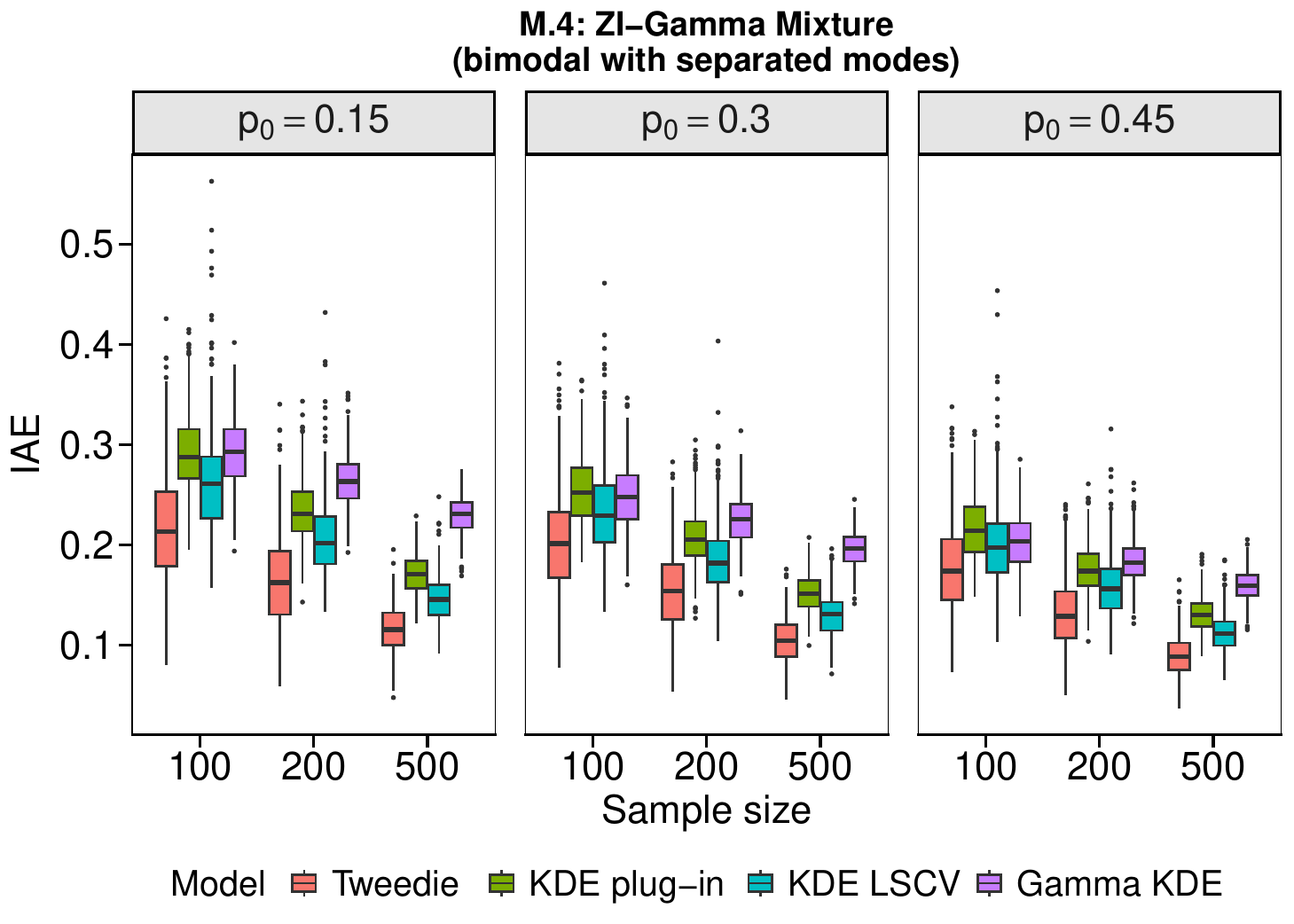}
\caption{Monte Carlo boxplots (500 replicates) of the integrated absolute error ($\mathrm{IAE}_{+}$) of the density estimators on the positive half-line $(0,\infty)$ under the four simulation scenarios. Each panel corresponds to a data-generating mechanism (M.1--M.4), with columns showing different values of the zero mass $p_0 \in \{0.15,0.30,0.45\}$ and sample sizes $n \in \{100,200,500\}$. The proposed Tweedie kernel estimator is compared with the Gaussian KDE using plug-in bandwidth selection, the Gaussian KDE with least-squares cross-validation (LSCV), and the Gamma KDE.}
\label{fig:simulationIAE}
\end{figure}

For completeness, the corresponding numerical summaries are reported in \ref{app:additional.tables}, where Tables~\ref{tab:Model1-2} and~\ref{tab:Model3-4} present the Monte Carlo means of the integrated squared error $\mathrm{ISE}_{+}$ and integrated absolute error $\mathrm{IAE}_{+}$ based on $R = 500$ independent replicates. Overall, the tabulated results are consistent with the graphical patterns in Figures~\ref{fig:simulationISE} and~\ref{fig:simulationIAE}. Across all scenarios, both error measures decrease with sample size for every method, and for a fixed $n$, they also tend to decrease as $p_0$ increases, reflecting the smaller scale of the positive component as more mass is concentrated at zero. Under the Tweedie design M.1, the proposed estimator performs consistently well, generally improving on the two Gaussian KDE competitors and remaining close to the Gamma KDE, which occasionally attains slightly smaller errors. Under the zero-inflated Gamma design M.2, the proposed estimator remains clearly preferable to the Gaussian KDE methods, although the Gamma KDE often performs best, as expected under this data-generating mechanism. The strongest advantage appears in the challenging misspecification setting M.3, where the positive component exhibits both a sharp boundary spike and a heavy right tail; in this case, the proposed estimator yields substantially smaller errors than the Gaussian and Gamma alternatives. In the bimodal setting M.4, the proposed estimator again achieves the smallest mean $\mathrm{ISE}_{+}$ and mean $\mathrm{IAE}_{+}$ in all settings, although the gains are somewhat less pronounced than in M.3.

\subsection{Sensitivity to power- and bandwidth-grid resolution}

In Table~\ref{tab:grid_sensitivity}, the labels ``coarse'', ``medium'', and ``fine'' denote the grid resolutions used in the profile search. Specifically, for the power parameter $p$, we use $N_p = 9$, $18$, and $27$ grid points, whereas for the bandwidth parameter $h$, we use $N_h = 20$, $40$, and $60$ grid points, corresponding to coarse, medium, and fine grids, respectively.

The results in Table~\ref{tab:grid_sensitivity} indicate that the proposed estimator is largely insensitive to the resolutions of the $p$- and $h$-grids used in the profile LSCV procedure. The reported values of the mean and standard deviation (SD) of $\mathrm{ISE}_{+}$ and $\mathrm{IAE}_{+}$ remain nearly unchanged across all nine grid combinations. In particular, the mean $\mathrm{ISE}_{+}$ remains essentially constant at $0.0079$ or $0.0080$; $\mathrm{SD}(\mathrm{ISE}_{+})$ varies only slightly, and the mean $\mathrm{IAE}_{+}$ changes negligibly. These results indicate that the tuning-parameter selection procedure is numerically stable and that coarse or medium grids are sufficient to deliver performance comparable to that obtained with finer grids.

\begin{table}[ht]
\centering
\caption{Effect of the power- and smoothing-parameter grid resolutions on estimator variability for Scenario M.4 (Zero-inflated Gamma Mixture, $n = 100$, $p_0 = 0.45$).}
\label{tab:grid_sensitivity}
\begin{tabular}{lcccccccccc}
\toprule[1.5pt]
& \multicolumn{3}{c}{$h$ grid coarse} & \multicolumn{3}{c}{$h$ grid medium} & \multicolumn{3}{c}{$h$ grid fine} \\
\cmidrule(lr){2-4} \cmidrule(lr){5-7} \cmidrule(lr){8-10}
$p$ grid & Coarse & Medium & Fine & Coarse & Medium & Fine & Coarse & Medium & Fine \\
\midrule[1.2pt]
mean $\mathrm{ISE}_{+}$ & 0.0079 &0.0079 & 0.0079& 0.0080 & 0.0079& 0.0079 & 0.0080& 0.0079&0.0079 \\
SD $\mathrm{ISE}_{+}$ & 0.0067 &0.0067 &0.0066 & 0.0064& 0.0064& 0.0064 & 0.0064& 0.0064&0.0064 \\
mean $\mathrm{IAE}_{+}$ &0.1773 &0.1777 &0.1778 & 0.1767&0.1765 & 0.1767& 0.1765& 0.1765& 0.1766\\
SD $\mathrm{IAE}_{+}$ &0.0460 &0.0447 &0.0443 & 0.0467&0.0459 & 0.0454& 0.0469& 0.0461& 0.0457\\
\bottomrule[1.5pt]
\end{tabular}
\end{table}

\subsection{Performance of the goodness-of-fit test}

We assess the finite-sample behavior of a density-based goodness-of-fit test under Scenario M.1. In this setting, the null model is the Tweedie distribution with mean parameter $\mu = 2$ and power parameter $p = 1.1$. Let $p_0 = 0.3$ denote the structural-zero probability used in the simulation design. Then the corresponding dispersion parameter is
\[
\phi = \frac{\mu^{\, 2 - p}}{(2 - p)(-\log p_0)} = \frac{2^{0.9}}{0.9(-\log p_0)},
\]
and the null parameter vector is
\[
\bb{\theta} = (\mu,\phi,p) = \left(2,\frac{2^{0.9}}{0.9(-\log 0.3)},1.1\right) \approx (2,1.722,1.1).
\]
Accordingly, the null and alternative hypotheses are
\[
\mathcal{H}_0: X \sim f_{\mathrm{Tw}}(\cdot; \bb{\theta}) \qquad \text{against} \qquad \mathcal{H}_1: X \not\sim f_{\mathrm{Tw}}(\cdot; \bb{\theta}).
\]
To measure departure from the null Tweedie model, we use a Cram\'er--von Mises-type statistic defined by the squared $L^2$-distance between the Tweedie kernel estimator and the null density on $(0, \infty)$, namely
\[
T_n = \int_0^{\infty} \bigl\{\widehat{g}_h(x) - f_{\mathrm{Tw}}(x; \bb{\theta})\bigr\}^2 \, \rd x.
\]
For practical implementation, the test is calibrated by Monte Carlo under $\mathcal{H}_0$. For each sample, we first compute the observed value of $T_n$. We then generate $B = 500$ Monte Carlo samples from the null Tweedie model, and for each generated sample recompute the statistic using the same smoothing procedure as for the original data. The critical value is taken as the empirical $0.95$-quantile of these $B$ simulated statistics. The null hypothesis is rejected whenever the observed value of $T_n$ exceeds this critical value. To study power, for each configuration we generate $R = 500$ independent Monte Carlo datasets and apply the above testing procedure to each dataset. The empirical rejection rate is then the proportion of rejections among these $R$ runs.

Figure~\ref{fig:GoF} reports the resulting rejection rates under the M.1 data-generating process. In the left panel, the alternative is indexed by the Tweedie power parameter $p$, with the null corresponding to $p = 1.1$. In the right panel, the alternative is indexed by the mean parameter $\mu$, with the null corresponding to $\mu = 2$. In both panels, the rejection rates under the null configuration are close to the nominal $5\%$ level. Away from the null, the rejection rates increase with the size of the departure, and the increase is sharper for larger sample sizes. For any fixed alternative, the rejection rates are higher when $n$ is larger, indicating that the test becomes more sensitive to departures from the null as the sample size increases. This behavior is consistent with the increasing sensitivity expected from the proposed procedure.

\medskip
\begin{figure}[H]
\centering
\includegraphics[trim={2mm 1mm 2mm 1mm}, clip, width=0.49\textwidth]{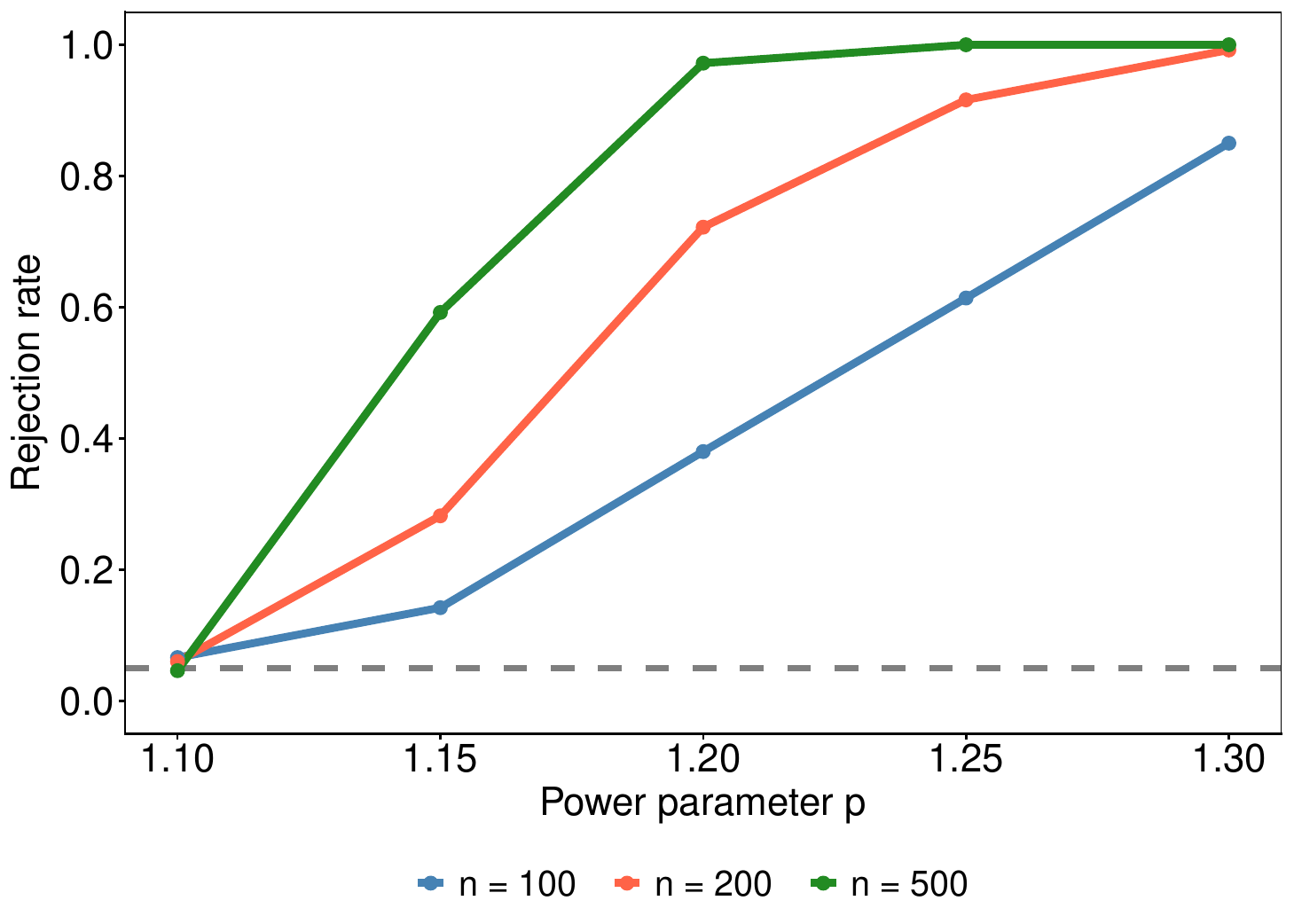}
\hfill
\includegraphics[trim={2mm 1mm 2mm 1mm}, clip, width=0.49\textwidth]{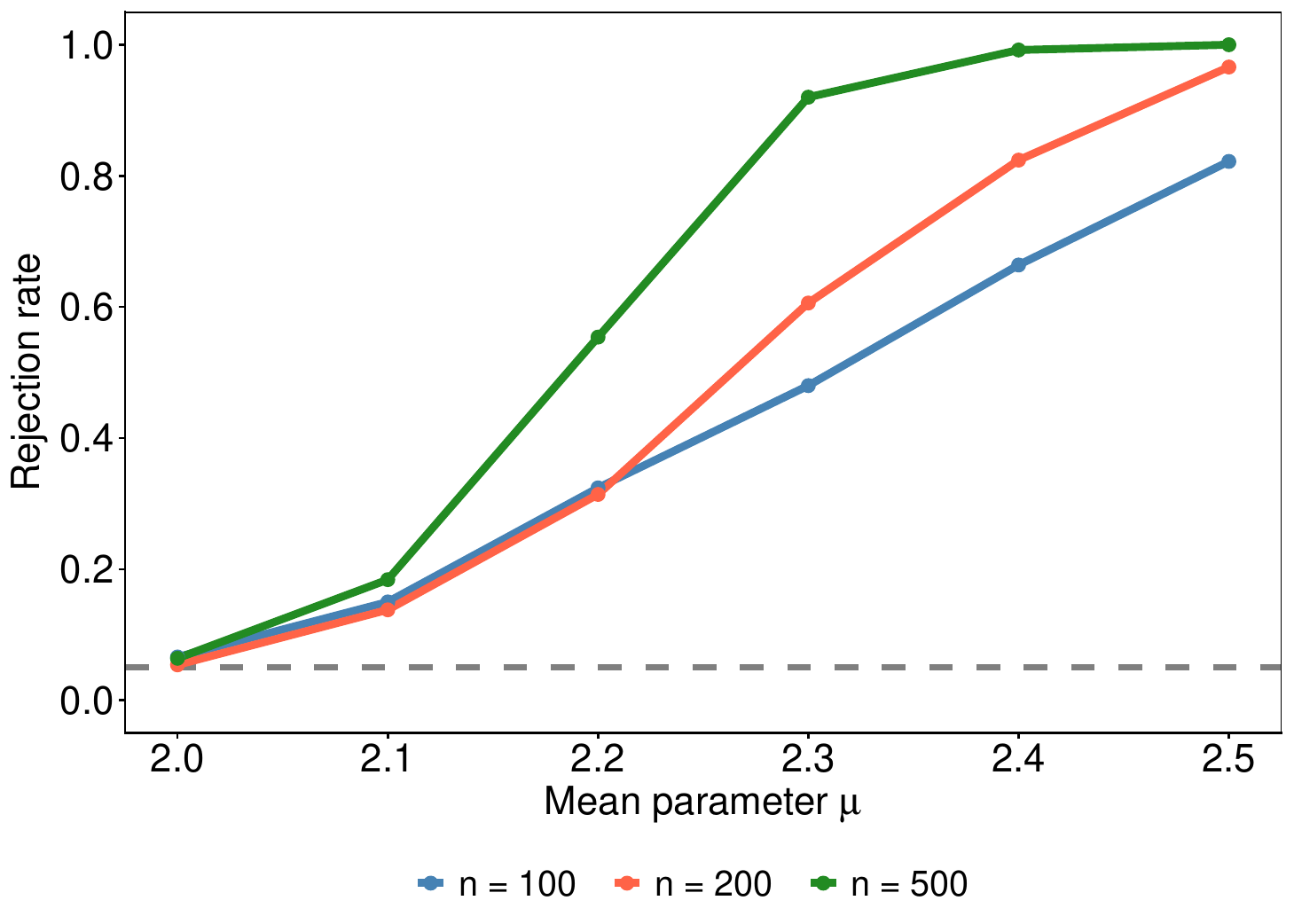}
\caption{Empirical rejection rates of the Tweedie estimator-based goodness-of-fit test under the M.1 data-generating process. The left panel plots rejection rates against the Tweedie power parameter $p$, and the right panel plots rejection rates against the mean parameter $\mu$, for sample sizes $n = 100, 200,$ and $500$. The dashed horizontal line marks the nominal $5\%$ significance level.}
\label{fig:GoF}
\end{figure}

\section{Real-data application}\label{sec:real.data.application}

For the empirical illustration, we analyze 732 mental health-related emergency department visits recorded in March 2023. The data were obtained from Nova Scotia Health's centralized Emergency Department Information System (EDIS), which uses standardized reporting procedures across hospitals. The data come from four principal emergency departments in the Central Zone of Nova Scotia in Canada: QEII Health Sciences Centre, Dartmouth General Hospital, Cobequid Community Health Centre, and Hants Community Hospital, which together account for most of the region's emergency care volume. The Central Zone serves approximately 461,000 residents, representing about 48\% of the provincial population \citep{Statcan2021zone4}. The outcome of interest is excess emergency department length of stay (LOS), defined as
\[
Y = \max(\text{LOS} - 4 \text{ hours},\, 0),
\]
where LOS denotes the total time from arrival to discharge. Thus, visits with $\text{LOS} \leq 4 ~\text{hours}$ have $Y = 0$, while visits exceeding 4 hours have positive values equal to the amount of delay beyond this threshold. The 4-hour threshold is used as a practically meaningful benchmark, as it aligns with an established emergency care performance target and distinguishes visits with no excess delay from those with prolonged stays. The resulting variable is semicontinuous, with a point mass at zero and a positively skewed continuous component, making it well suited for illustrating the proposed method.

Figure~\ref{fig:histogram} highlights two salient features of the excess length-of-stay distribution for mental health visits. The left panel shows a substantial point mass at zero together with a strongly right-skewed positive component, indicating that many visits do not exceed the 4-hour benchmark, while positive exceedances are concentrated near zero and taper off gradually over a long right tail. In the right panel, the density on the positive half-line peaks near the boundary and then decreases smoothly, a shape that is well captured by the proposed Tweedie kernel estimator without imposing a strict parametric form. Among the fitted parametric competitors, the zero-inflated Lognormal model appears to provide the closest overall approximation, followed by the zero-inflated Gamma model, whereas the Tweedie MLE shows the largest discrepancy. This ranking is also supported by the integrated squared distances from the proposed nonparametric estimate: 0.022376 for the zero-inflated Lognormal fit, 0.080832 for the zero-inflated Gamma fit, and 0.242282 for the Tweedie MLE. Overall, the figure indicates pronounced zero inflation and substantial positive skewness, while also showing that a flexible nonparametric estimator can serve as a useful benchmark for assessing the adequacy of the parametric fits, particularly in the near-boundary region of the positive component.

\medskip
\begin{figure}[H]
\centering
\includegraphics[trim={2mm 1mm 2mm 1mm}, clip, width=0.49\textwidth]{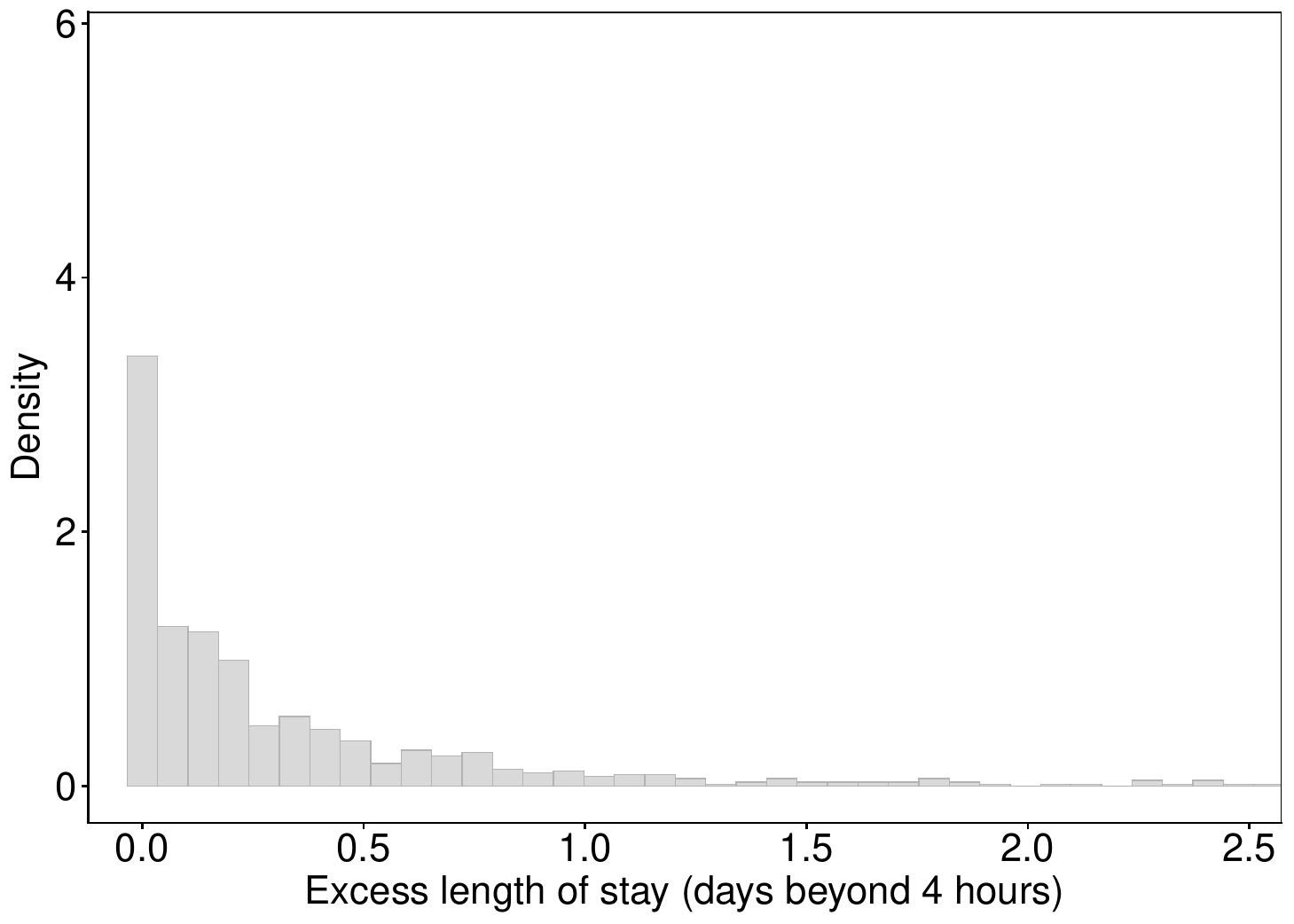}
\hfill
\includegraphics[trim={2mm 1mm 2mm 1mm}, clip, width=0.49\textwidth]{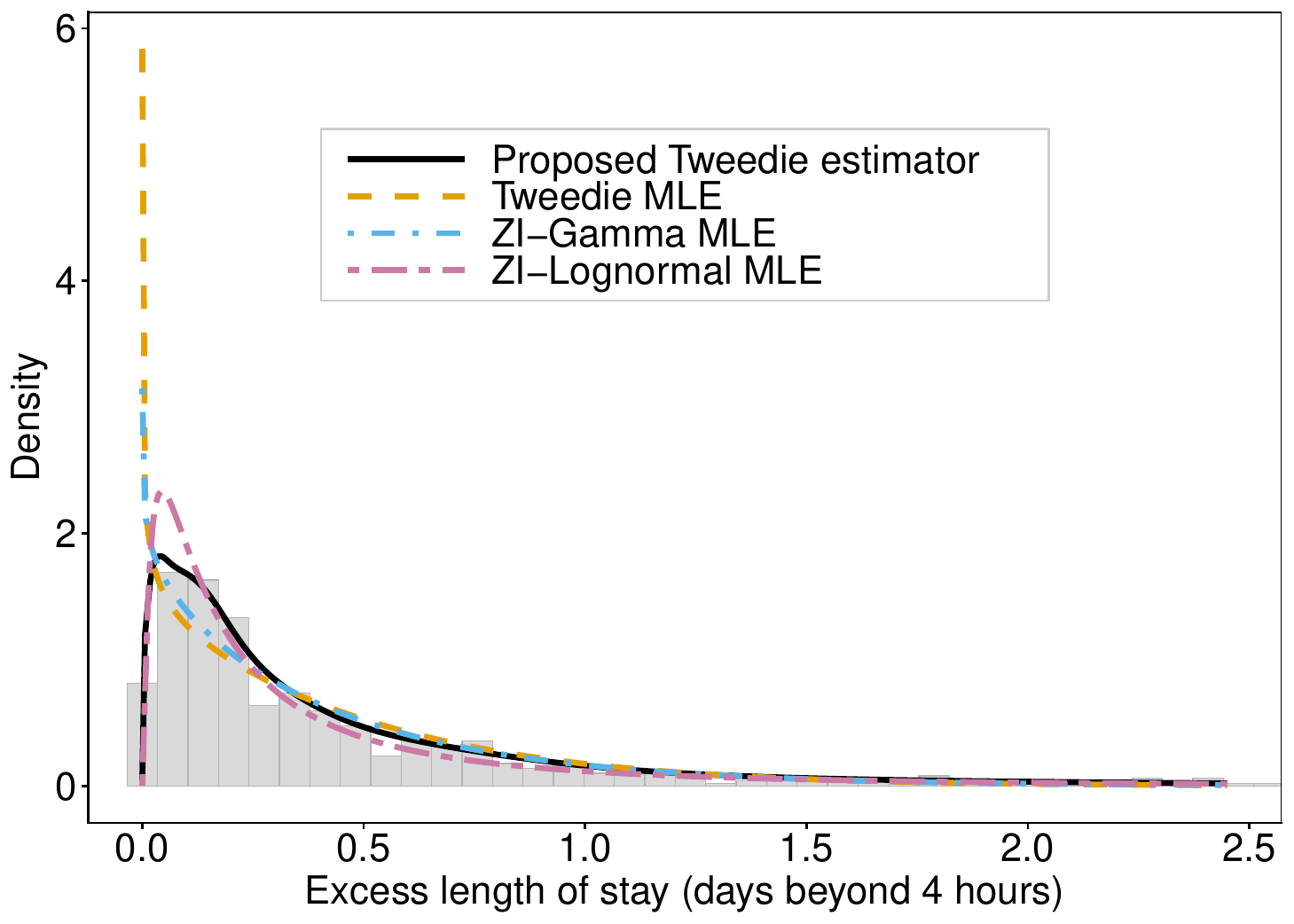}
\caption{Emergency department excess length of stay (days beyond 4 hours) for mental health visits. Left panel: histogram of all observations, including the zero mass. Right panel: density estimates on the positive half-line. The solid black curve is the proposed Tweedie kernel estimator; the dashed, dot-dashed, and two-dash curves correspond to fitted Tweedie, zero-inflated Gamma, and zero-inflated Lognormal models, respectively. The background histogram displays the positive observations only, scaled to the marginal target function $g(x) = (1 - p_0)f(x)$.}
\label{fig:histogram}
\end{figure}

\section{Concluding remarks}\label{sec:conclusion}

In this paper, we studied nonparametric estimation for semicontinuous distributions by introducing a Tweedie kernel estimator for the full mixed density on $[0,\infty)$. The proposed construction exploits the compound Poisson--Gamma structure of Tweedie laws with power parameter $p\in(1,2)$, so that the point mass at zero and the continuous positive component are handled within a single smoothing framework. At the origin, the estimator reduces exactly to the empirical proportion of zeros, whereas on $(0,\infty)$ it provides asymmetric positive-support smoothing based on the full sample. This unified treatment is especially relevant in the near-zero region, where boundary effects are strongest and where the interaction between exact zeros and small positive observations makes the density particularly difficult to recover.

The theoretical analysis shows that the estimator admits tractable pointwise asymptotics. In particular, pointwise bias and variance expansions are obtained, asymptotic $\mathrm{MSE}$ and $\mathrm{MISE}$ formulae are derived together with the corresponding optimal bandwidth rates, and asymptotic normality is established under regularity conditions adapted to the behavior of the positive component near the origin. A profile least-squares cross-validation procedure is also proposed for the joint selection of the bandwidth and the Tweedie power parameter, thereby allowing the amount of smoothing and the kernel shape to be chosen in a data-driven manner.

The finite-sample results indicate that the proposed estimator is consistently competitive and can offer clear advantages in challenging semicontinuous settings, particularly when the positive component exhibits a pronounced boundary spike, a heavy right tail, or substantial misspecification of the positive-component shape. The emergency department application further illustrates that the estimator can serve as a flexible nonparametric benchmark for assessing parametric alternatives and for revealing structural features of the data that may be difficult to capture with more restrictive models.

Several directions remain open. One concerns the theoretical analysis of fully data-driven tuning, including sharper asymptotic results for the joint selection of the bandwidth and the Tweedie power parameter. Another concerns extensions beyond the pointwise framework considered here, for example to stronger uniform and $L_1$ results~\citep{BouezmarniRolin2007}, to multivariate or regression-type semicontinuous settings, and to inference for functionals of the mixed density. A further direction is to extend the proposed estimator to censored survival data. More broadly, the results here suggest that Tweedie-based smoothing provides a promising framework for nonparametric analysis of distributions that combine nonnegative support with an intrinsic mixed discrete--continuous structure.

\appendix

\begin{appendices}

\renewcommand{\thesection}{Appendix~\Alph{section}}
\section{Auxiliary results}\label{app:auxiliary.results}
\renewcommand{\thesection}{\Alph{section}}

Throughout this section, let $k_h(\, \cdot \, ; x)$ denote the Lebesgue subdensity of the absolutely continuous part of the Tweedie law on $(0,\infty)$ with mean $\mu = x$, dispersion $\phi = h$, and variance function $V(\mu) = \mu^p$, namely
\begin{equation}\label{eq:k.h}
k_h(t; x) = e^{-\lambda_x} \sum_{j=1}^{\infty} \frac{\lambda_x ^j}{j!} \, \frac{t^{j\alpha - 1}e^{-t/\beta_x }}{\Gamma(j\alpha)\beta_x ^{j\alpha}}, \qquad t > 0,
\end{equation}
where, as before,
\begin{equation}\label{eq:alpha.lambda.beta}
\lambda_x = \frac{x^{2 - p}}{h(2 - p)},\qquad \alpha = \frac{2 - p}{p - 1},\qquad \beta_x = h(p - 1)x^{p - 1}.
\end{equation}

Lemma~\ref{lem:stirling} is a direct consequence of the small-dispersion asymptotics for exponential dispersion models and the associated saddlepoint approximation for the Tweedie distribution; see~\citet[][Sections~2.3 and~4.1]{Jorgensen1987ExponentialDispersion} and~\citet[Section~3.5, p.~103]{Jorgensen1997theory}.

\begin{lemma}[Uniform saddlepoint approximation for Tweedie kernels]\label{lem:stirling}
Fix $\eta\in(0,1)$ and $p\in(1,2)$, and let $I_{\eta} = [\eta,\eta^{-1}]$. Define the unit deviance
\[
d_p(u,x) = 2\int_x ^u \frac{u - t}{t^p} \, \rd t = 2\left\{\frac{u^{2 - p}}{(1 - p)(2 - p)} - \frac{u x^{1 - p}}{1 - p} + \frac{x^{2 - p}}{2 - p}\right\}, \qquad u > 0,\ x > 0.
\]
Then, for every fixed compact set $K\subseteq (0,\infty)$,
\[
\sup_{x\in I_{\eta}}\sup_{u\in K} \left| \frac{k_h(u; x)}{(2\pi h \, u^p)^{-1/2}\exp\left(-d_p(u,x)/(2h)\right)} - 1 \right| \longrightarrow 0, \qquad h\to 0.
\]
\end{lemma}

\begin{corollary}[Gaussian approximation for Tweedie kernels]\label{cor:local-gaussian}
Under the assumptions of Lemma~\ref{lem:stirling}, for every $M > 0$,
\[
\sup_{x\in I_{\eta}}\sup_{\substack{u > 0 \\ |u - x| \leq M \sqrt{h}}} \left| \frac{k_h(u; x)}{(2\pi h x^p)^{-1/2} \exp\left(-(u - x)^2 / (2h x^p)\right)} - 1 \right| \longrightarrow 0, \qquad h\to 0.
\]
Equivalently,
\[
k_h(u; x) = \frac{1}{\sqrt{2\pi h x^p}} \exp\left(-\frac{(u - x)^2}{2h x^p}\right)\{1 + \oo(1)\}, \qquad h\to 0,
\]
uniformly for $x\in I_{\eta}$ and $u > 0$ such that $|u - x|\leq M \sqrt{h}$.
\end{corollary}

\begin{lemma}\label{lem:wright-bound}
Let $\alpha > 0$, and define
\begin{equation}\label{eq:def.H.alpha}
H_{\alpha}(A) = \sum_{j=1}^{\infty} \frac{A^j}{j! \, \Gamma(j\alpha)}, \qquad A\geq 0.
\end{equation}
Then there exist positive constants $C_{\alpha}$ and $\rho_{\alpha}$ such that
\[
H_{\alpha}(A)\leq C_{\alpha} \, A \, \exp\bigl(\rho_{\alpha} A^{1/(1 + \alpha)}\bigr), \qquad A\ge0.
\]
\end{lemma}

\begin{proof}[Proof of Lemma~\ref{lem:wright-bound}]
Throughout the proof, $C_{\alpha}$ denotes a positive constant depending only on $\alpha$, whose values may change from line to line. For $A\in[0,1]$, since $A^j\leq A$ for all $j\ge1$,
\[
H_{\alpha}(A) \leq A\sum_{j=1}^{\infty} \frac{1}{j! \, \Gamma(j\alpha)} \leq C_{\alpha} A.
\]
Now let $A\ge1$. By Stirling's formula, there exists a constant $c_{\alpha} > 0$ such that, for all integers $j\ge1$,
\[
j! \, \Gamma(j\alpha) \geq c_{\alpha} \, \alpha^{\alpha j}e^{-(1 + \alpha)j}j^{(1 + \alpha)j}.
\]
Hence
\[
\frac{A^j}{j! \, \Gamma(j\alpha)}
\leq C_{\alpha} \left(\frac{e^{1 + \alpha}A}{\alpha^{\alpha} j^{1 + \alpha}}\right)^j
= C_{\alpha} e^{-j} \exp\left\{j \log\left(\frac{e^{2 + \alpha}A}{\alpha^{\alpha}}\right) - (1 + \alpha) j \log j\right\}.
\]
Note that, for $y > 0$,
\[
\sup_{t > 0} \bigl\{t \log y - (1 + \alpha) \, t \log t\bigr\} = (1 + \alpha)e^{-1}y^{1/(1 + \alpha)}.
\]
Applying this with $y = e^{2 + \alpha}A/\alpha^{\alpha}$ yields
\[
\frac{A^j}{j! \, \Gamma(j\alpha)} \leq C_{\alpha} e^{-j} \exp\bigl(\rho_{\alpha} A^{1/(1 + \alpha)}\bigr), \qquad j\ge1,\ A\ge1,
\]
for some $\rho_{\alpha} > 0$. Summing over $j\ge1$ gives
\[
H_{\alpha}(A) \leq C_{\alpha} \exp\bigl(\rho_{\alpha} A^{1/(1 + \alpha)}\bigr) \leq C_{\alpha} A \exp\bigl(\rho_{\alpha} A^{1/(1 + \alpha)}\bigr), \qquad A\ge1.
\]
Combining the cases $A\in[0,1]$ and $A\ge1$ proves the claim.
\end{proof}

\begin{proposition}[Near-zero bound for the Tweedie kernel]\label{prop:near-zero}
Fix $x\in(0,\infty)$ and $p\in(1,2)$, and let $\alpha = (2 - p)/(p - 1)$. Then there exist constants $\varepsilon = \varepsilon(x,p)\in (0,x/2)$, $a_0 = a_0(p) > 0$, $c_0 = c_0(x,p) > 0$, and $C_0 = C_0(p) > 0$ such that, for all sufficiently small $h$,
\[
k_h(t; x) \leq C_0 \, h^{-a_0} \, e^{-c_0/h} \, t^{\alpha - 1}, \qquad 0 < t \leq \varepsilon.
\]
\end{proposition}

\begin{proof}[Proof of Proposition~\ref{prop:near-zero}]
By \eqref{eq:k.h} and the definition of $H_{\alpha}$ in \eqref{eq:def.H.alpha},
\[
k_h(t; x) = e^{-\lambda_x - t/\beta_x } \, t^{-1} \, H_{\alpha}(A_h(t)), \qquad A_h(t) = \lambda_x \bigg(\frac{t}{\beta_x}\bigg)^{\alpha}, \qquad t > 0.
\]
Applying Lemma~\ref{lem:wright-bound} yields
\[
k_h(t; x) \leq C_{\alpha} \, e^{-\lambda_x - t/\beta_x } \, t^{-1} \, A_h(t) \, \exp\bigl(\rho_{\alpha} A_h(t)^{1/(1 + \alpha)}\bigr),
\]
for some constant $C_{\alpha} > 0$ depending only on $\alpha$. Since $A_h(t) = \lambda_x \beta_x ^{-\alpha}t^{\alpha}$, this becomes
\begin{equation}\label{eq:near-zero-pre}
k_h(t; x) \leq C_{\alpha} \, \lambda_x \beta_x ^{-\alpha}t^{\alpha - 1} \exp\left\{-\lambda_x - \frac{t}{\beta_x } + \rho_{\alpha} A_h(t)^{1/(1 + \alpha)}\right\}.
\end{equation}
A direct calculation using \eqref{eq:alpha.lambda.beta} gives
\[
A_h(t)^{1/(1 + \alpha)} = c_p \, \frac{t^{2 - p}}{h}, \qquad c_p = (2 - p)^{-(p - 1)}(p - 1)^{-(2 - p)}.
\]
Therefore the exponential factor in \eqref{eq:near-zero-pre} is
\[
\exp\left\{-\frac{x^{2 - p}}{(2 - p)h} - \frac{t}{(p - 1)hx^{p - 1}} + \rho_{\alpha} c_p \, \frac{t^{2 - p}}{h}\right\}.
\]
Choose $\varepsilon\in(0,x/2)$ so small that $\rho_{\alpha} c_p \, \varepsilon^{2 - p} \leq x^{2 - p}/\{2(2 - p)\}$. Then, for every $0 < t\leq \varepsilon$,
\[
-\frac{x^{2 - p}}{(2 - p)h} - \frac{t}{(p - 1)hx^{p - 1}} + \rho_{\alpha} c_p \, \frac{t^{2 - p}}{h} \leq - \frac{x^{2 - p}}{2(2 - p)h}.
\]
Moreover, $\lambda_x \beta_x ^{-\alpha} = (2 - p)^{-1} (p - 1)^{-\alpha} h^{-(1 + \alpha)}$. Hence
\[
k_h(t; x) \leq C_0 \, h^{-(1 + \alpha)}e^{-c_0/h}t^{\alpha - 1}, \qquad 0 < t\leq \varepsilon,
\]
with $a_0 = 1 + \alpha = 1/(p - 1)$, $c_0 = x^{2 - p}/\{2(2 - p)\}$, and some $C_0 > 0$ depending only on $p$, as claimed.
\end{proof}

\begin{proposition}[Away-from-zero exponential upper bound]\label{prop:away-zero}
Fix $p\in(1,2)$, $x\in(0,\infty)$ and $\varepsilon\in(0,x/2)$. Then there exists a constant $C_1 = C_1(x,p,\varepsilon) > 0$ such that, for all sufficiently small $h$,
\[
k_h(t; x)\leq C_1 \, h^{-1/2} \exp\left\{-\frac{d_p(t,x)}{4h}\right\}, \qquad t\geq \varepsilon.
\]
\end{proposition}

\begin{proof}[Proof of Proposition~\ref{prop:away-zero}]
Split the proof into a compact part and a tail part.

\begin{enumerate}[(i)]
\item The compact interval $[\varepsilon,L]$. Let $L > x\vee1$ be chosen sufficiently large, as specified below. On the compact set $[\varepsilon,L]\subseteq (0,\infty)$, Lemma~\ref{lem:stirling} yields
\[
k_h(t; x) = (2\pi h \, t^p)^{-1/2} \exp\left\{-\frac{d_p(t,x)}{2h}\right\}\{1 + \oo(1)\},
\]
uniformly in $t\in[\varepsilon,L]$. Hence, for all sufficiently small $h$,
\[
k_h(t; x) \leq 2(2\pi h \, t^p)^{-1/2} \exp\left\{-\frac{d_p(t,x)}{2h}\right\} \leq C_{\varepsilon,L} \, h^{-1/2} \exp\left\{-\frac{d_p(t,x)}{4h}\right\},
\]
for some constant $C_{\varepsilon,L} > 0$, uniformly for $t\in[\varepsilon,L]$.

\item The tail region $t\geq L$. By Lemma~\ref{lem:wright-bound}, the argument used in the proof of Proposition~\ref{prop:near-zero} gives
\begin{equation}\label{eq:tail-pre}
k_h(t; x) \leq C \, h^{-(1 + \alpha)} \, t^{\alpha - 1} \exp\left\{-\frac{q(t)}{h}\right\}, \qquad t > 0,
\end{equation}
for some constant $C > 0$, where
\[
q(t) = \frac{x^{2 - p}}{2 - p} + \frac{t}{(p - 1)x^{p - 1}} - b_{\alpha} t^{2 - p}, \qquad b_{\alpha} = \rho_{\alpha} c_p > 0.
\]
On the other hand, from the explicit formula for the Tweedie deviance,
\[
\frac{d_p(t,x)}{2} = \frac{x^{2 - p}}{2 - p} + \frac{t}{(p - 1)x^{p - 1}} - \frac{t^{2 - p}}{(p - 1)(2 - p)}.
\]
Since $2 - p\in(0,1)$, the linear term dominates the sublinear term $t^{2 - p}$ as $t\to \infty$. Therefore, choose $L$ sufficiently large so that
\[
q(t) - \gamma t \geq \frac13 \, d_p(t,x), \qquad t\geq L,
\]
where $\gamma = 1/\{4(p - 1)x^{p - 1}\}$. Also, for some constant $C_{\gamma} > 0$, for $t\geq L$ and $h\le1$,
\[
t^{\alpha - 1}\leq C_{\gamma} e^{\gamma t}\leq C_{\gamma} e^{\gamma t/h}.
\]
Substituting this into \eqref{eq:tail-pre}, we obtain
\[
k_h(t; x) \leq C \, h^{-(1 + \alpha)} \exp\left\{-\frac{q(t) - \gamma t}{h}\right\} \leq C \, h^{-(1 + \alpha)} \exp\left\{-\frac{d_p(t,x)}{3h}\right\}, \qquad t\geq L.
\]
Since $d_p(t,x)\geq d_p(L,x) > 0$ for $t\geq L$, the exponential term absorbs the extra polynomial factor in $h$: for all sufficiently small $h$,
\[
h^{-(1 + \alpha)} \exp\left\{-\frac{d_p(t,x)}{3h}\right\} \leq \widetilde{C} \, h^{-1/2} \exp\left\{-\frac{d_p(t,x)}{4h}\right\}, \qquad t\geq L,
\]
for some constant $\smash{\widetilde{C} > 0}$.
\end{enumerate}
Combining this with the bound on $[\varepsilon,L]$ completes the proof.
\end{proof}

\begin{lemma}\label{lem:moment-condition}
Fix $p\in(1,2)$, $x\in(0,\infty)$, and let $\delta > 0$. Suppose that Assumption~\ref{ass:smooth} holds and that Assumption~\ref{ass:boundary} holds with $r = 2 + \delta$. Then
\[
\EE\{K_h(X; x)^{2 + \delta}\} = \OO(h^{-(1 + \delta)/2}), \qquad h\to 0.
\]
\end{lemma}

\begin{proof}[Proof of Lemma~\ref{lem:moment-condition}]
Write $r = 2 + \delta$. Throughout the proof, $C$ and $c$ denote positive constants, depending only on fixed quantities, whose values may change from line to line. Since $X$ is semicontinuous,
\[
\PP(X = 0) = p_0, \qquad X\mid(X > 0)\sim f(u) \, \rd u,
\]
and therefore
\[
\EE\{K_h(X; x)^r\} = p_0 K_h(0; x)^r + (1 - p_0)\int_0^{\infty} k_h(u; x)^r f(u) \, \rd u.
\]
Because
\[
K_h(0; x) = \pi_h(x) = \exp\left\{-\frac{x^{2 - p}}{(2 - p)h}\right\},
\]
the atomic term is exponentially small: $p_0 K_h(0; x)^r = \OO(e^{-c/h}) = \oo(h^m)$ for every $m > 0$.

It remains to bound the integral term. Fix $\varepsilon\in(0,x/2)$ small enough that Assumption~\ref{ass:boundary} holds on $(0,\varepsilon)$ with exponent $r$ and Proposition~\ref{prop:near-zero} applies on $(0,\varepsilon)$. Also choose $\delta_0\in(0,x/2)$ so small that $[x - \delta_0,x + \delta_0]$ is contained in the smoothness neighborhood of $x$ from Assumption~\ref{ass:smooth}, and choose $c_{\ast} > 0$ sufficiently small so that
\[
d_p(u,x)\geq c_{\ast} (u - x)^2, \qquad |u - x|\leq \delta_0;
\]
this follows from the local expansion
\[
d_p(u,x) = \frac{(u - x)^2}{x^p} + \oo\bigl((u - x)^2\bigr) \qquad\text{as }u\to x.
\]
For all sufficiently small $h$, assume $\sqrt{h}<\delta_0$. Then
\[
\int_0^{\infty} k_h(u; x)^r f(u) \, \rd u = \left(\int_{\mathcal{R}_0} + \int_{\mathcal{R}_1} + \int_{\mathcal{R}_2} + \int_{\mathcal{R}_3}\right) k_h(u; x)^r f(u) \, \rd u,
\]
where
\[
\begin{aligned}
\mathcal{R}_0 &= (0,\varepsilon), \\
\mathcal{R}_1 &= \{u\in (0,\infty) : |u - x|\leq \sqrt{h}\}, \\
\mathcal{R}_2 &= \{u\in (0,\infty) : \sqrt{h} < |u - x| \leq \delta_0\}, \\
\mathcal{R}_3 &= \{u\in [\varepsilon,\infty) : |u - x| > \delta_0\}.
\end{aligned}
\]
Let us bound these four terms separately.
\begin{enumerate}[(i)]
\item The near-zero region $\mathcal{R}_0 = (0,\varepsilon)$. By Proposition~\ref{prop:near-zero} and Assumption~\ref{ass:boundary},
\[
\int_{\mathcal{R}_0} k_h(u; x)^r f(u) \, \rd u \leq C h^{-a_0 r} e^{-r c_0/h} \int_0^{\varepsilon} u^{r(\alpha - 1)} f(u) \, \rd u = \OO(h^{-a_0 r} e^{-r c_0/h}) = \oo(h^{-(r - 1)/2}).
\]
%%%
\item The local region $\mathcal{R}_1 = \{u\in (0,\infty) : |u - x|\leq \sqrt{h}\}$. Since $f$ is continuous at the fixed point $x > 0$, it is bounded on some neighborhood of $x$ by Assumption~\ref{ass:smooth}. By Corollary~\ref{cor:local-gaussian}, uniformly for $u\in \mathcal{R}_1$,
\[
k_h(u; x) = \frac{1}{\sqrt{2\pi h x^p}} \exp\left(-\frac{(u - x)^2}{2h x^p}\right)\{1 + \oo(1)\}.
\]
Therefore, for all sufficiently small $h$,
\[
\int_{\mathcal{R}_1} k_h(u; x)^r f(u) \, \rd u \leq C h^{-r/2} \int_{|u - x|\leq \sqrt{h}} \exp\left\{-\frac{r(u - x)^2}{C h}\right\} \, \rd u.
\]
With the change of variables $u = x + \sqrt{h} \, t$,
\[
\int_{\mathcal{R}_1} k_h(u; x)^r f(u) \, \rd u \leq C h^{-r/2} \sqrt{h} \int_{|t|\leq 1} e^{-c t^2} \, \rd t = \OO(h^{-(r - 1)/2}).
\]
%%%
\item The intermediate region $\mathcal{R}_2 = \{u\in (0,\infty) : \sqrt{h} < |u - x| \leq \delta_0\}$. Since $\delta_0 < x/2$, one has $u\geq x - \delta_0\geq x/2 > \varepsilon$ on the region $\mathcal{R}_2$. Hence Proposition~\ref{prop:away-zero} applies and gives
\[
k_h(u; x) \leq C h^{-1/2} \exp\left\{-\frac{d_p(u,x)}{4h}\right\}.
\]
Using the quadratic lower bound on $d_p(u,x)$ for $|u - x|\leq \delta_0$ yields
\[
k_h(u; x) \leq C h^{-1/2} \exp\left\{-\frac{c_{\ast} (u - x)^2}{4h}\right\}.
\]
Since $f$ is bounded on $[x - \delta_0,x + \delta_0]$,
\[
\int_{\mathcal{R}_2} k_h(u; x)^r f(u) \, \rd u \leq C h^{-r/2} \int_{ \sqrt{h} < |u - x|\leq \delta_0} \exp\left\{-\frac{r c_{\ast} (u - x)^2}{4h}\right\} \, \rd u.
\]
Again setting $u = x + \sqrt{h} \, t$, we obtain
\[
\int_{\mathcal{R}_2} k_h(u; x)^r f(u) \, \rd u \leq C h^{-r/2} \sqrt{h} \int_{|t| > 1} e^{-c t^2} \, \rd t = \OO(h^{-(r - 1)/2}).
\]
%%%
\item The far region $\mathcal{R}_3 = \{u\in [\varepsilon,\infty) : |u - x| > \delta_0\}$. On the set $\mathcal{R}_3$, the function $d_p(u,x)$ is continuous, strictly positive, and tends to $\infty$ as $u\to \infty$. Let $c_{\delta_0,\varepsilon} = \inf_{u\in \mathcal{R}_3} d_p(u,x) > 0$. By Proposition~\ref{prop:away-zero},
\[
k_h(u; x) \leq C h^{-1/2} \exp\left\{-\frac{d_p(u,x)}{4h}\right\} \leq C h^{-1/2} e^{-c_{\delta_0,\varepsilon}/(4h)}, \qquad u\in \mathcal{R}_3.
\]
Therefore,
\[
\int_{\mathcal{R}_3} k_h(u; x)^r f(u) \, \rd u \leq C h^{-r/2} e^{-r c_{\delta_0,\varepsilon}/(4h)} \int_{\mathcal{R}_3} f(u) \, \rd u = \oo(h^{-(r - 1)/2}).
\]
\end{enumerate}
Combining the four bounds, we conclude that
\[
\int_0^{\infty} k_h(u; x)^r f(u) \, \rd u = \OO(h^{-(r - 1)/2}).
\]
Since the atomic term is exponentially small, it follows that
\[
\EE\{K_h(X; x)^r\} = \OO(h^{-(r - 1)/2}),
\]
which proves the lemma.
\end{proof}

\renewcommand{\thesection}{Appendix~\Alph{section}}
\section{Proofs of the main results}\label{app:proofs.main.results}
\renewcommand{\thesection}{\Alph{section}}

\subsection{Proof of \texorpdfstring{\hyperref[thm:bias_tkde]{Theorem~\ref{thm:bias_tkde}}}{Theorem~\ref{thm:bias_tkde}}}

Fix $x\in (0,\infty)$, let $\alpha = (2 - p)/(p - 1)$, let $r = 2 + \delta$ be the exponent in Assumption~\ref{ass:boundary}, and let $Z\sim \mathrm{Tw}_p(\mu = x,\phi = h)$. Since $X_1,\dots,X_n$ are i.i.d.,
\[
\EE\{\widehat{g}_h(x)\}
= \EE\{K_h(X; x)\}
= p_0 \, K_h(0; x) + (1 - p_0)\int_{0}^{\infty} K_h(u; x) \, f(u) \, \rd u .
\]
For $p\in(1,2)$, the Tweedie law of $Z$ satisfies
\[
\PP(Z = 0) = \pi_h(x) = \exp\left\{-\frac{x^{2 - p}}{(2 - p)h}\right\},
\]
and admits a Lebesgue subdensity $k_h(\cdot \, ; x)$ on $(0,\infty)$ such that, for any measurable $\varphi$ for which the expectation exists,
\[
\EE\{\varphi(Z)\} = \varphi(0) \, \pi_h(x) + \int_{0}^{\infty}\varphi(u) \, k_h(u; x) \, \rd u.
\]
Hence $K_h(0; x) = \pi_h(x)$, and for $u > 0$ the quantity $K_h(u; x)$ is interpreted as the subdensity value $k_h(u; x)$. Extending $f$ to $[0, \infty)$ by defining $f(0) = 0$, we have
\begin{equation}\label{eq:Eg_decomp}
\begin{aligned}
\EE\{\widehat{g}_h(x)\}
&= p_0 \, \pi_h(x) + (1 - p_0)\int_{0}^{\infty} k_h(u; x) \, f(u) \, \rd u \\
&= p_0 \, \pi_h(x) + (1 - p_0) \, \EE\{f(Z)\}.
\end{aligned}
\end{equation}
Since $x$ is fixed,
\[
p_0 \, \pi_h(x) = p_0\exp\left\{-\frac{x^{2 - p}}{(2 - p)h}\right\} = \oo(h), \qquad h\to 0.
\]

Choose $\varepsilon_0\in(0,x/2\wedge 1)$ small enough that Assumption~\ref{ass:boundary} holds on $(0,\varepsilon_0]$ and Proposition~\ref{prop:near-zero} applies on $(0,\varepsilon_0]$. Next, fix $\eta > 0$, and choose $\rho\in(0,(x - \varepsilon_0)/2)$ so small that $f\in C^2([x - \rho,x + \rho])$ and
\[
\omega_x (\rho) : = \sup_{|u - x|\leq \rho}|f''(u) - f''(x)| \leq \eta.
\]
Set
\[
B_0 = (0,\varepsilon_0], \qquad B_1 = [x - \rho,x + \rho], \qquad B_2 = (\varepsilon_0,\infty)\setminus B_1.
\]
Because $f(0) = 0$,
\[
\EE\{f(Z)\} = \EE\{f(Z)\ind(Z\in B_0)\} + \EE\{f(Z)\ind(Z\in B_1)\} + \EE\{f(Z)\ind(Z\in B_2)\}.
\]

Let us first control the near-zero contribution. By Proposition~\ref{prop:near-zero},
\[
\EE\{f(Z)\ind(Z\in B_0)\} = \int_0^{\varepsilon_0} k_h(u; x) \, f(u) \, \rd u \leq C_0 h^{-a_0}e^{-c_0/h}\int_0^{\varepsilon_0} u^{\alpha - 1}f(u) \, \rd u.
\]
The integral on the right is finite. Indeed, if $\alpha\ge1$, then $u^{\alpha - 1}\leq \varepsilon_0^{\alpha - 1}$ on $(0,\varepsilon_0]$, whereas if $\alpha < 1$, then $r > 1$ and $\varepsilon_0\le1$ give $u^{\alpha - 1}\leq u^{r(\alpha - 1)}$ for $0 < u\leq \varepsilon_0$, so Assumption~\ref{ass:boundary} implies $\int_0^{\varepsilon_0}u^{\alpha - 1}f(u) \, \rd u < \infty$. Therefore
\begin{equation}\label{eq:bias-near-zero}
\EE\{f(Z)\ind(Z\in B_0)\} = \oo(h).
\end{equation}

Next, by Proposition~\ref{prop:away-zero},
\[
\EE\{f(Z)\ind(Z\in B_2)\} = \int_{B_2} k_h(u; x) \, f(u) \, \rd u \leq C_1 h^{-1/2}\int_{B_2} \exp\left\{-\frac{d_p(u,x)}{4h}\right\}f(u) \, \rd u.
\]
Since $d_p(\cdot,x)$ is continuous on $[\varepsilon_0,\infty)$, satisfies $d_p(u,x) > 0$ for $u\neq x$, and $d_p(u,x)\to \infty$ as $u\to \infty$,
\[
m_{\rho}: = \inf_{u\in B_2} d_p(u,x) > 0.
\]
Hence
\begin{equation}\label{eq:bias-away}
\EE\{f(Z)\ind(Z\in B_2)\} \leq C_1 h^{-1/2}e^{-m_{\rho}/(4h)}\int_{B_2}f(u) \, \rd u = \oo(h).
\end{equation}

On $B_1$, Taylor's theorem gives
\[
f(u) = f(x) + f'(x)(u - x) + \frac{1}{2} f''(x)(u - x)^2 + R_x (u), \qquad u\in B_1,
\]
where
\[
|R_x (u)| \leq \frac{1}{2} \, \omega_x (\rho)(u - x)^2.
\]
Writing $\Delta_h = Z - x$, one has
\begin{equation}\label{eq:bias-local}
\begin{aligned}
\EE\{f(Z)\ind(Z\in B_1)\}
&= f(x)\PP(Z\in B_1) + f'(x) \, \EE\{\Delta_h\ind(Z\in B_1)\} \\
& \qquad + \frac{1}{2} f''(x) \, \EE\{\Delta_h^2\ind(Z\in B_1)\} + \EE\{R_x (Z)\ind(Z\in B_1)\}.
\end{aligned}
\end{equation}

For $m = 0,1,2$, Proposition~\ref{prop:near-zero} and the bound $|\Delta_h|\leq x$ on $B_0$ yield
\begin{equation}\label{eq:nearzero-moments}
\EE\{|\Delta_h|^m\ind(Z\in B_0)\} \leq x^m\int_0^{\varepsilon_0} k_h(u; x) \, \rd u \leq C h^{-a_0}e^{-c_0/h}\int_0^{\varepsilon_0}u^{\alpha - 1} \, \rd u = \oo(h).
\end{equation}

Because $2 - p\in(0,1)$, the explicit formula for $d_p(u,x)$ shows that $d_p(u,x)$ grows linearly as $u\to \infty$. Hence there exist $L > (x + \rho)\vee1$ and $c_L > 0$ such that
\[
d_p(u,x)\geq c_L u, \qquad u\geq L.
\]
Then, for $m = 0,1,2$, Proposition~\ref{prop:away-zero} gives
\begin{align*}
\EE\{|\Delta_h|^m\ind(Z\in B_2)\}
&= \int_{B_2} |u - x|^m k_h(u; x) \, \rd u \\
&\leq C_1 h^{-1/2}\int_{B_2\cap[\varepsilon_0,L]} |u - x|^m \exp\left\{-\frac{d_p(u,x)}{4h}\right\} \, \rd u \\
& \qquad + C_1 h^{-1/2}\int_{L}^{\infty} u^m \exp\left\{-\frac{c_L u}{4h}\right\} \, \rd u.
\end{align*}
Since the closure of $B_2\cap[\varepsilon_0,L]$ is compact and does not contain $x$,
\[
m_{\rho,L}: = \inf_{u\in \overline{B_2\cap[\varepsilon_0,L]}} d_p(u,x) > 0,
\]
so the first term is $\OO\bigl(h^{-1/2}e^{-m_{\rho,L}/(4h)}\bigr) = \oo(h)$. For the second term, the change of variables $v = c_L u/(4h)$ gives
\[
h^{-1/2}\int_{L}^{\infty} u^m \exp\left\{-\frac{c_L u}{4h}\right\} \, \rd u = \OO\bigl(h^{m + 1/2}e^{-c_L L/(8h)}\bigr) = \oo(h).
\]
Consequently,
\begin{equation}\label{eq:far-moments}
\EE\{|\Delta_h|^m\ind(Z\in B_2)\} = \oo(h), \qquad m = 0,1,2.
\end{equation}

Using
\[
\PP(Z = 0) = \pi_h(x) = \oo(h), \qquad \EE(\Delta_h) = 0, \qquad \EE(\Delta_h^2) = \Var(Z) = h x^p,
\]
and \eqref{eq:nearzero-moments}--\eqref{eq:far-moments}, one has
\[
\PP(Z\in B_1) = 1 - \oo(h), \qquad \EE\{\Delta_h\ind(Z\in B_1)\} = \oo(h), \qquad \EE\{\Delta_h^2\ind(Z\in B_1)\} = h x^p + \oo(h).
\]
Moreover,
\[
\bigl|\EE\{R_x (Z)\ind(Z\in B_1)\}\bigr| \leq \frac{1}{2} \, \omega_x (\rho) \, \EE\{\Delta_h^2\ind(Z\in B_1)\} \leq \frac{1}{2} \, \eta \, \{h x^p + \oo(h)\}.
\]
Substituting these bounds into \eqref{eq:bias-local} yields
\[
\EE\{f(Z)\ind(Z\in B_1)\} = f(x) + \frac{1}{2} h x^p f''(x) + \OO(\eta h) + \oo(h).
\]
Combining this with \eqref{eq:bias-near-zero} and \eqref{eq:bias-away} yields
\[
\EE\{f(Z)\} = f(x) + \frac{1}{2} h x^p f''(x) + \OO(\eta h) + \oo(h).
\]
Since $\eta > 0$ was arbitrary,
\begin{equation}\label{eq:Ef_bias}
\EE\{f(Z)\} = f(x) + \frac{1}{2} h x^p f''(x) + \oo(h).
\end{equation}
Finally, combining \eqref{eq:Eg_decomp} and \eqref{eq:Ef_bias}, and using $g(x) = (1 - p_0)f(x)$ and $g''(x) = (1 - p_0)f''(x)$ for $x > 0$, one has
\[
\EE\{\widehat{g}_h(x)\} - g(x) = \frac{1}{2} \, h \, x^p g''(x) + \oo(h),
\]
which completes the proof. \qed

\subsection{Proof of \texorpdfstring{\hyperref[thm:var]{Theorem~\ref{thm:var}}}{Theorem~\ref{thm:var}}}

Recall that
\[
\widehat{g}_h(x) = \frac{1}{n} \sum_{i=1}^n K_h(X_i; x),
\]
where $K_h(\cdot; x)$ is the Tweedie kernel with mean parameter $\mu = x$, dispersion $h$, and power $p\in(1,2)$. Set
\[
\alpha = \frac{2 - p}{p - 1}.
\]
Throughout the proof, $C$ denotes a positive constant, depending only on fixed quantities and independent of $h$ and $M$, whose value may change from line to line. Since $X_1,\dots,X_n$ are i.i.d.,
\begin{equation}\label{eq:var-basic-polished}
\Var\{\widehat{g}_h(x)\} = \frac{1}{n} \Var\{K_h(X; x)\} = \frac{1}{n} \big(\EE[K_h(X; x)^2] - \EE[K_h(X; x)]^2\big).
\end{equation}
For $p\in(1,2)$, the Tweedie kernel has an atom at $0$ of size
\[
\pi_h(x) = \exp\left\{-\frac{x^{2 - p}}{(2 - p)h}\right\},
\]
and a Lebesgue subdensity $k_h(\cdot; x)$ on $(0,\infty)$. Thus, if $Z$ has this Tweedie law, then for any measurable $\varphi$ for which the expectation exists,
\[
\EE\{\varphi(Z)\} = \varphi(0)\pi_h(x) + \int_0^{\infty} \varphi(u)k_h(u; x) \, \rd u.
\]
We interpret $K_h(u; x) = k_h(u; x)$ as this subdensity value for $u > 0$, and $K_h(0; x) = \pi_h(x)$. Since $x > 0$ is fixed, $\pi_h(x)$ is exponentially small in $h^{-1}$; in particular, for every $m > 0$,
\[
\pi_h(x) = \oo(h^m), \qquad h\to 0.
\]
Hence all terms involving $\pi_h(x)$ are negligible relative to $h^{-1/2}$. Then
\begin{equation}\label{eq:EKh2-start}
\EE[K_h(X; x)^2] = p_0 \, \pi_h(x)^2 + (1 - p_0)\int_0^{\infty} k_h(u; x)^2 f(u) \, \rd u.
\end{equation}
The first term is $\oo(h^m)$ for every $m > 0$, so it remains to evaluate the integral term.

Fix $x\in(0,\infty)$. Let $r = 2 + \delta$ be the exponent in Assumption~\ref{ass:boundary}, and choose $\varepsilon\in(0,x/2)$ so that
\[
\int_0^{\varepsilon} u^{r(\alpha - 1)} f(u) \, \rd u < \infty
\]
and the conclusion of Proposition~\ref{prop:near-zero} holds on $(0,\varepsilon]$. Choose $\rho\in(0,x/2)$ sufficiently small so that $[x - \rho,x + \rho]$ is contained in the smoothness neighborhood of $x$ from Assumption~\ref{ass:smooth}. For fixed $M > 0$, and for all sufficiently small $h$ such that $M\sqrt{h} < \rho$, write
\[
\int_0^{\infty} k_h(u; x)^2 f(u) \, \rd u = \left(\int_{|u - x|\leq M \sqrt{h}} + \int_{M \sqrt{h} < |u - x|\leq \rho} + \int_{|u - x| > \rho}\right) k_h(u; x)^2 f(u) \, \rd u.
\]
\begin{enumerate}[(i)]
\item Let us first consider the local region $|u - x|\leq M \sqrt{h}$. Since $f$ is continuous at $x$, uniformly for $|u - x|\leq M \sqrt{h}$, $f(u) = f(x) + \oo(1)$. Moreover, by Corollary~\ref{cor:local-gaussian},
\[
k_h(u; x) = \frac{1}{\sqrt{2\pi h x^p}} \exp\left\{-\frac{(u - x)^2}{2h x^p}\right\}\{1 + \oo(1)\}, \qquad h\to 0,
\]
uniformly for $|u - x|\leq M \sqrt{h}$. Hence
\[
k_h(u; x)^2 = \frac{1}{2\pi h x^p} \exp\left\{-\frac{(u - x)^2}{h x^p}\right\}\{1 + \oo(1)\},
\]
uniformly on the same region, and therefore
\[
\int_{|u - x|\leq M \sqrt{h}} k_h(u; x)^2 f(u) \, \rd u = \frac{f(x)}{2\pi h x^p} \int_{|u - x|\leq M \sqrt{h}} \exp\left\{-\frac{(u - x)^2}{h x^p}\right\} \, \rd u + \oo(h^{-1/2}).
\]
With the substitution $u = x + \sqrt{h} \, t$,
\[
\int_{|u - x|\leq M \sqrt{h}} k_h(u; x)^2 f(u) \, \rd u = \frac{f(x)}{2\pi \sqrt{h} \, x^p} \int_{|t|\leq M}\exp\left(-\frac{t^2}{x^p}\right) \, \rd t + \oo(h^{-1/2}).
\]

\item Next consider the intermediate region $M \sqrt{h} < |u - x|\leq \rho$. By the local expansion of $d_p(u,x)$, there exists $c_1 > 0$ such that
\[
d_p(u,x)\geq c_1 (u - x)^2, \qquad |u - x|\leq \rho.
\]
This follows from the fact that
\[
d_p(u,x) = \frac{(u - x)^2}{x^p} + \oo\bigl((u - x)^2\bigr) \qquad\text{as }u\to x.
\]
By Lemma~\ref{lem:stirling}, uniformly for $u\in[x - \rho,x + \rho]$,
\[
k_h(u; x) = (2\pi h \, u^p)^{-1/2} \exp\left\{-\frac{d_p(u,x)}{2h}\right\}\{1 + \oo(1)\}.
\]
Hence, for all sufficiently small $h$,
\[
\int_{M \sqrt{h} < |u - x|\leq \rho} k_h(u; x)^2 f(u) \, \rd u \leq C h^{-1} \int_{M \sqrt{h} < |u - x|\leq \rho} \exp\left\{-\frac{c_1 (u - x)^2}{h}\right\} \, \rd u.
\]
Again setting $u = x + \sqrt{h} \, t$, one has
\[
\int_{M \sqrt{h} < |u - x|\leq \rho} k_h(u; x)^2 f(u) \, \rd u \leq C h^{-1/2}\int_{|t| > M} e^{-c_1 t^2} \, \rd t.
\]
Therefore,
\[
\limsup_{h\to 0} h^{1/2}\int_{M \sqrt{h} < |u - x|\leq \rho} k_h(u; x)^2 f(u) \, \rd u \leq C\int_{|t| > M} e^{-c_1 t^2} \, \rd t,
\]
and the right-hand side tends to $0$ as $M\to \infty$.

\item It remains to handle the region $|u - x| > \rho$. Since $\varepsilon < x/2$ and $\rho < x/2$, we have $\varepsilon < x - \rho$, so $(0,\varepsilon]\subset\{u:|u - x| > \rho\}$. Choose $L > x + \rho$. Then $[L,\infty)\subset\{u:|u - x| > \rho\}$, and hence
\[
\int_{|u - x| > \rho} k_h(u; x)^2 f(u) \, \rd u = \left(\int_0^{\varepsilon} + \int_{\substack{\varepsilon < u\leq L\\ |u - x| > \rho}} + \int_L^{\infty}\right) k_h(u; x)^2 f(u) \, \rd u.
\]
where $L > x + \rho$ is chosen sufficiently large, as specified below. By H\"older's inequality,
\[
\left(\int_0^{\varepsilon} u^{2(\alpha - 1)}f(u) \, \rd u\right)^{1/2} \leq \left(\int_0^{\varepsilon} u^{r(\alpha - 1)}f(u) \, \rd u\right)^{1/r} \left(\int_0^{\varepsilon} f(u) \, \rd u\right)^{1/2 - 1/r} < \infty.
\]
Therefore, by Proposition~\ref{prop:near-zero},
\[
\int_0^{\varepsilon} k_h(u; x)^2 f(u) \, \rd u \leq C h^{-2a_0} e^{-2c_0/h} \int_0^{\varepsilon} u^{2(\alpha - 1)} f(u) \, \rd u = \oo(h^{-1/2}).
\]

On the compact set $A_{\rho,\varepsilon,L} = \{u\in[\varepsilon,L]:|u - x|\geq \rho\}$, the deviance $d_p(u,x)$ is continuous and strictly positive, so $c_{\rho,\varepsilon,L} = \inf_{u\in A_{\rho,\varepsilon,L}} d_p(u,x) > 0$. Applying Lemma~\ref{lem:stirling} uniformly on the compact set $A_{\rho,\varepsilon,L}$,
\[
\int_{\substack{\varepsilon < u\leq L\\ |u - x| > \rho}} k_h(u; x)^2 f(u) \, \rd u \leq C h^{-1} e^{-c_{\rho,\varepsilon,L}/h} = \oo(h^{-1/2}).
\]

For the tail term, by Proposition~\ref{prop:away-zero},
\[
k_h(u; x)^2 \leq \left(C h^{-1/2}\exp\left\{-\frac{d_p(u,x)}{4h}\right\}\right)^2 \leq C h^{-1}\exp\left\{-\frac{d_p(u,x)}{2h}\right\}, \qquad u\geq L,
\]
which is valid for sufficiently small $h$. Since $d_p(u,x)\to \infty$ as $u\to \infty$, choose $L$ so large that
\[
d_p(u,x)\geq 2c_L > 0, \qquad u\geq L.
\]
Therefore
\[
\int_L^{\infty} k_h(u; x)^2 f(u) \, \rd u \leq C h^{-1} e^{-c_L/h}\int_L^{\infty} f(u) \, \rd u = \oo(h^{-1/2}).
\]
Hence
\[
\int_{|u - x| > \rho} k_h(u; x)^2 f(u) \, \rd u = \oo(h^{-1/2}).
\]
\end{enumerate}
Combining the three pieces, for each fixed $M > 0$,
\[
\int_0^{\infty} k_h(u; x)^2 f(u) \, \rd u = \frac{f(x)}{2\pi \sqrt{h} \, x^p} \int_{|t|\leq M}\exp\left(-\frac{t^2}{x^p}\right) \, \rd t + \oo(h^{-1/2}) + R_h(M),
\]
where $R_h(M)$ is the contribution from the intermediate region and satisfies
\[
\limsup_{h\to 0} h^{1/2}|R_h(M)| \leq C\int_{|t| > M} e^{-c_1 t^2} \, \rd t.
\]
Letting $M\to \infty$ and using
\[
\int_{-\infty}^{\infty} \exp\left(-\frac{t^2}{x^p}\right) \, \rd t = \sqrt{\pi x^p},
\]
we obtain
\[
\int_0^{\infty} k_h(u; x)^2 f(u) \, \rd u = \frac{f(x)}{2\sqrt{\pi} \, x^{p/2}} \, h^{-1/2} + \oo(h^{-1/2}).
\]
Substituting this into \eqref{eq:EKh2-start} gives
\[
\EE[K_h(X; x)^2] = (1 - p_0)\frac{f(x)}{2\sqrt{\pi} \, x^{p/2}} \, h^{-1/2} + \oo(h^{-1/2}) = \frac{g(x)}{2\sqrt{\pi} \, x^{p/2}} \, h^{-1/2} + \oo(h^{-1/2}).
\]
Finally, from the already established bias expansion,
\[
\EE[K_h(X; x)] = \EE[\widehat{g}_h(x)] = g(x) + \OO(h),
\]
so
\[
\EE[K_h(X; x)]^2 = g(x)^2 + \OO(h) = \OO(1) = \oo(h^{-1/2}).
\]
Combining this with \eqref{eq:var-basic-polished} yields
\[
\Var\{\widehat{g}_h(x)\} = \frac{1}{n} \left\{\frac{g(x)}{2\sqrt{\pi} \, x^{p/2}} \, h^{-1/2} + \oo(h^{-1/2})\right\},
\]
which proves the result. \qed

\subsection{Proof of \texorpdfstring{\hyperref[thm:asymp-normal]{Theorem~\ref{thm:asymp-normal}}}{Theorem~\ref{thm:asymp-normal}}}

Define
\[
Y_{i,n} = K_h(X_i; x) - \EE\{K_h(X; x)\}, \qquad i = 1,\dots,n.
\]
Then $\EE(Y_{i,n}) = 0$, and for each $n$, the variables $Y_{1,n},\dots,Y_{n,n}$ are i.i.d. Moreover,
\[
\widehat{g}_h(x) - \EE\{\widehat{g}_h(x)\} = \frac{1}{n} \sum_{i=1}^n Y_{i,n}.
\]
Set
\[
\sigma^2(x) = \frac{g(x)}{2\sqrt{\pi} \, x^{p/2}}, \qquad s_n^2 = \sum_{i=1}^n \Var(Y_{i,n}) = n\Var(Y_{1,n}).
\]
By Theorem~\ref{thm:var},
\[
\Var\{\widehat{g}_h(x)\} = \frac{1}{n} \{\sigma^2(x) \, h^{-1/2} + \oo(h^{-1/2})\}, \qquad n\to \infty,
\]
and hence
\[
\Var(Y_{1,n}) = n \, \Var\{\widehat{g}_h(x)\} = \sigma^2(x) \, h^{-1/2} + \oo(h^{-1/2}).
\]
Next, let $r = 2 + \delta$, where $\delta > 0$ is as in Assumption~\ref{ass:boundary}. By the inequality
\[
|a - b|^r \leq 2^{r - 1}(|a|^r + |b|^r),
\]
one has
\[
\EE|Y_{1,n}|^r \leq 2^{r - 1} \left(\EE\left\{K_h(X; x)^r\right\} + \bigl|\EE\{K_h(X; x)\}\bigr|^r\right).
\]
By Lemma~\ref{lem:moment-condition}, $\EE\{K_h(X; x)^r\} = \OO(h^{-(r - 1)/2})$, while Theorem~\ref{thm:bias_tkde} implies
\[
\EE\{K_h(X; x)\} = \EE\{\widehat{g}_h(x)\} = g(x) + \OO(h) = \OO(1).
\]
Hence
\begin{equation}\label{eq:rth-moment-centered}
\EE|Y_{1,n}|^r = \OO(h^{-(r - 1)/2}).
\end{equation}

First suppose that $g(x) > 0$. Then $\sigma^2(x) > 0$, so the variance expansion above can be rewritten as
\[
\Var(Y_{1,n}) = \sigma^2(x) \, h^{-1/2}\{1 + \oo(1)\}, \qquad n\to \infty,
\]
and therefore
\begin{equation}\label{eq:sn2-clt}
s_n^2 = n \, \sigma^2(x) \, h^{-1/2}\{1 + \oo(1)\}.
\end{equation}
Now verify the Lyapunov condition:
\[
\frac{1}{s_n^r}\sum_{i=1}^n \EE|Y_{i,n}|^r = \frac{n \, \EE|Y_{1,n}|^r}{s_n^r}.
\]
Using \eqref{eq:sn2-clt} and \eqref{eq:rth-moment-centered},
\[
\frac{1}{s_n^r}\sum_{i=1}^n \EE|Y_{i,n}|^r
= \frac{n \, \OO(h^{-(r - 1)/2})} {\big(n \, \sigma^2(x) \, h^{-1/2}\{1 + \oo(1)\}\big)^{r/2}}
= \OO(n^{1 - r/2} \, h^{-(r - 1)/2 + r/4}).
\]
Since $r = 2 + \delta$, $1 - r/2 = -\delta/2$, $-(r-1)/2 + r/4 = -\delta/4$, we have
\[
\frac{1}{s_n^r}\sum_{i=1}^n \EE|Y_{i,n}|^r = \OO(n^{-\delta/2}h^{-\delta/4}) = \OO((n h^{1/2})^{-\delta/2}) \longrightarrow 0,
\]
because $n h^{1/2}\to \infty$ by Assumption~\ref{ass:parameter2}. Thus the Lyapunov condition holds.

Therefore, by the Lyapunov central limit theorem for triangular arrays~\citep[Theorem~27.3]{Billingsley1995},
\[
\frac{1}{s_n}\sum_{i=1}^n Y_{i,n} \stackrel{d}{\longrightarrow} \mathcal{N}(0,1).
\]
Now
\[
n^{1/2} h^{1/4} \big(\widehat{g}_h(x) - \EE\{\widehat{g}_h(x)\}\big) = \frac{h^{1/4}}{\sqrt n}\sum_{i=1}^n Y_{i,n} = \frac{h^{1/4}s_n}{\sqrt n} \, \frac{1}{s_n}\sum_{i=1}^n Y_{i,n}.
\]
By \eqref{eq:sn2-clt},
\[
\frac{h^{1/4}s_n}{\sqrt n} = \sqrt{h^{1/2}\Var(Y_{1,n})} \longrightarrow \sqrt{\sigma^2(x)}.
\]
Hence Slutsky's theorem gives
\[
n^{1/2} h^{1/4} \big(\widehat{g}_h(x) - \EE\{\widehat{g}_h(x)\}\big) \stackrel{d}{\longrightarrow} \mathcal{N}\bigl(0,\sigma^2(x)\bigr),
\]
where $\sigma^2(x) = g(x)/(2\sqrt{\pi}x^{p/2})$. If $g(x) = 0$, then $\sigma^2(x) = 0$, and the variance expansion above yields
\[
\Var(Y_{1,n}) = \oo(h^{-1/2}).
\]
Hence
\[
\Var\left[n^{1/2} h^{1/4} \big(\widehat{g}_h(x) - \EE\{\widehat{g}_h(x)\}\big)\right] = n h^{1/2}\Var\{\widehat{g}_h(x)\} = h^{1/2}\Var(Y_{1,n}) \longrightarrow 0.
\]
Since the random variable inside the variance has mean zero, it follows that
\[
n^{1/2} h^{1/4} \big(\widehat{g}_h(x) - \EE\{\widehat{g}_h(x)\}\big) \longrightarrow 0 \qquad \text{in }L^2,
\]
and hence converges in probability and in distribution to the degenerate distribution $\mathcal{N}(0,0)$.

Therefore, in all cases,
\[
n^{1/2} h^{1/4} \big(\widehat{g}_h(x) - \EE\{\widehat{g}_h(x)\}\big) \stackrel{d}{\longrightarrow} \mathcal{N}\left(0,\frac{g(x)}{2\sqrt{\pi} \, x^{p/2}}\right).
\]
Finally, by Theorem~\ref{thm:bias_tkde},
\[
\EE\{\widehat{g}_h(x)\} - g(x) = \frac{1}{2} \, h x^p g''(x) + \oo(h),
\]
so
\[
n^{1/2} h^{1/4} \bigl(\EE\{\widehat{g}_h(x)\} - g(x)\bigr) = \frac{1}{2} x^p g''(x) n^{1/2} h^{5/4} + \oo(n^{1/2} h^{5/4}).
\]

Moreover, if $n h^{5/2}\to \lambda\in[0,\infty)$, then
\[
n^{1/2} h^{1/4} \bigl(\EE\{\widehat{g}_h(x)\} - g(x)\bigr) \longrightarrow \frac{1}{2} x^p g''(x)\sqrt{\lambda}.
\]
Combining this with the centered limit above and applying Slutsky's theorem,
\[
n^{1/2} h^{1/4} \big(\widehat{g}_h(x) - g(x)\big) \stackrel{d}{\longrightarrow} \mathcal{N}\left(\frac{1}{2} x^p g''(x)\sqrt{\lambda}, \frac{g(x)}{2\sqrt{\pi} \, x^{p/2}}\right).
\]
This completes the proof. \qed

\renewcommand{\thesection}{Appendix~\Alph{section}}
\section{Additional simulation tables}\label{app:additional.tables}
\renewcommand{\thesection}{\Alph{section}}

\begin{table}[H]
\centering
\caption{Monte Carlo performance of density estimators on the positive half-line $(0,\infty)$ under simulation scenarios M.1 and M.2. For each sample size $n$ and zero mass $p_0$, the Monte Carlo means of the integrated squared error ($\mathrm{ISE}_{+}$) and integrated absolute error ($\mathrm{IAE}_{+}$) based on $500$ simulation replicates are reported. Results are shown for the proposed Tweedie kernel estimator, Gaussian KDE with plug-in and least-squares cross-validation (LSCV) bandwidth selection, and Gamma KDE.}
\label{tab:Model1-2}
\setlength{\tabcolsep}{2.7pt}
\renewcommand{\arraystretch}{0.90}
\scalebox{0.85}{
\begin{tabular}{c c c *{8}{r}}
\toprule[1.5pt]
\multirow{2}{*}{Scenario} & \multirow{2}{*}{$n$} & \multirow{2}{*}{$p_0$} & \multicolumn{2}{c}{Tweedie estimator} & \multicolumn{2}{c}{KDE plug-in} & \multicolumn{2}{c}{KDE LSCV} & \multicolumn{2}{c}{Gamma KDE} \\
\cmidrule(lr){4-5}
\cmidrule(lr){6-7}
\cmidrule(lr){8-9}
\cmidrule(lr){10-11}
& & & Mean $\mathrm{ISE}_{+}$ & Mean $\mathrm{IAE}_{+}$ & Mean $\mathrm{ISE}_{+}$ & Mean $\mathrm{IAE}_{+}$ & Mean $\mathrm{ISE}_{+}$ & Mean $\mathrm{IAE}_{+}$ & Mean $\mathrm{ISE}_{+}$ & Mean $\mathrm{IAE}_{+}$ \\
\midrule[1pt]
\multirow{9}{*}{M.1} & \multirow{3}{*}{100} & 0.15 & 0.0091 & 0.1684 & 0.0110 & 0.1873 & 0.0120 & 0.1984 & 0.0084 & 0.1592 \\
& & 0.30 & 0.0065 & 0.1500 & 0.0070 & 0.1680 & 0.0080 & 0.1781 & 0.0057 & 0.1435 \\
& & 0.45 &0.0041 & 0.1315 & 0.0039 & 0.1408 & 0.0046 & 0.1502 & 0.0034 & 0.1234 \\
\cmidrule(lr){2-11}
& \multirow{3}{*}{200} & 0.15 & 0.0052 & 0.1264 & 0.0074 & 0.1510 & 0.0071 & 0.1537 & 0.0054 & 0.1259 \\
& & 0.30 & 0.0038 & 0.1136 & 0.0044 & 0.1315 & 0.0047 & 0.1373 & 0.0037 & 0.1136 \\
& & 0.45 & 0.0023 & 0.0985 & 0.0023 & 0.1080 & 0.0027 & 0.1147 & 0.0021 & 0.0958 \\
\cmidrule(lr){2-11}
& \multirow{3}{*}{500} & 0.15 & 0.0026 & 0.0890 & 0.0041 & 0.1110 & 0.0036 & 0.1108 & 0.0032 & 0.0939 \\
& & 0.30 & 0.0017 & 0.0777 & 0.0022 & 0.0927 & 0.0022 & 0.0949& 0.0022 & 0.0863 \\
& & 0.45 & 0.0010 & 0.0672 & 0.0011 & 0.0760 & 0.0013 & 0.0797 & 0.0012 &0.0722 \\
\midrule[1pt]
\multirow{9}{*}{M.2} & \multirow{3}{*}{100} & 0.15 & 0.0971 & 0.1725 & 0.0844 & 0.1742 & 0.1041 & 0.1987 & 0.0768 & 0.1440 \\
& & 0.30 & 0.0759 &0.1529 & 0.0661 & 0.1549 & 0.0792 & 0.1730 & 0.0614 & 0.1311 \\
& & 0.45 & 0.0583 & 0.1339 & 0.0491 & 0.1349 & 0.0599 & 0.1491 & 0.0461 & 0.1152 \\
\cmidrule(lr){2-11}
& \multirow{3}{*}{200} & 0.15 & 0.0549 & 0.1303 & 0.0577 & 0.1404 & 0.0653 & 0.1586 & 0.0497 & 0.1102 \\
& & 0.30 & 0.0448 & 0.1154 & 0.0444 & 0.1244 & 0.0516 & 0.1410 & 0.0390 & 0.1001 \\
& & 0.45 & 0.0311 & 0.0977 & 0.0317 & 0.1063 & 0.0374 & 0.1206 & 0.0282 & 0.0873 \\
\cmidrule(lr){2-11}
& \multirow{3}{*}{500} & 0.15 & 0.0269 & 0.0923 & 0.0349 & 0.1046 & 0.0375 & 0.1199 & 0.0328 & 0.0787 \\
& & 0.30 & 0.0216 & 0.0828 & 0.0271 & 0.0931 & 0.0294 & 0.1059 & 0.0246 & 0.0708 \\
& & 0.45 & 0.0160 & 0.0707 & 0.0189 & 0.0791 & 0.0210 & 0.0897 & 0.0170 & 0.0614 \\
\bottomrule[1.5pt]
\end{tabular}}
\end{table}

\begin{table}[H]
\centering
\caption{Monte Carlo performance of density estimators on the positive half-line $(0,\infty)$ under simulation scenarios M.3 and M.4. For each sample size $n$ and zero mass $p_0$, the Monte Carlo means of the integrated squared error ($\mathrm{ISE}_{+}$) and integrated absolute error ($\mathrm{IAE}_{+}$) based on $500$ simulation replicates are reported. Results are shown for the proposed Tweedie kernel estimator, Gaussian KDE with plug-in and least-squares cross-validation (LSCV) bandwidth selection, and Gamma KDE.}
\label{tab:Model3-4}
\setlength{\tabcolsep}{2.7pt}
\renewcommand{\arraystretch}{0.90}
\scalebox{0.85}{
\begin{tabular}{c c c *{8}{r}}
\toprule[1.5pt]
\multirow{2}{*}{Scenario} & \multirow{2}{*}{$n$} & \multirow{2}{*}{$p_0$} & \multicolumn{2}{c}{Tweedie estimator} & \multicolumn{2}{c}{KDE plug-in} & \multicolumn{2}{c}{KDE LSCV} & \multicolumn{2}{c}{Gamma KDE} \\
\cmidrule(lr){4-5}
\cmidrule(lr){6-7}
\cmidrule(lr){8-9}
\cmidrule(lr){10-11}
& & & Mean $\mathrm{ISE}_{+}$ & Mean $\mathrm{IAE}_{+}$ & Mean $\mathrm{ISE}_{+}$ & Mean $\mathrm{IAE}_{+}$ & Mean $\mathrm{ISE}_{+}$ & Mean $\mathrm{IAE}_{+}$ & Mean $\mathrm{ISE}_{+}$ & Mean $\mathrm{IAE}_{+}$ \\
\midrule[1pt]
\multirow{9}{*}{M.3} & \multirow{3}{*}{100} & 0.15 & 0.0334 & 0.2401 & 0.2234 & 0.5763 & 0.0516 & 0.4813 & 0.1568 & 0.5227 \\
& & 0.30 & 0.0286 & 0.2259 & 0.1558 & 0.4880 & 0.0403 & 0.4233 & 0.1112 & 0.4450 \\
& & 0.45 & 0.0205 & 0.1995 & 0.0995 & 0.3970 & 0.0286 & 0.3573 & 0.0730 & 0.3646 \\
\cmidrule(lr){2-11}
& \multirow{3}{*}{200} & 0.15 & 0.0190 & 0.1812 & 0.1975 & 0.5161 & 0.0384 & 0.3960 & 0.1383 & 0.4773 \\
& & 0.30 & 0.0150 & 0.1689 & 0.1395 & 0.4389 & 0.0283 & 0.3470 & 0.0973 & 0.4033 \\
& & 0.45 & 0.0114 & 0.1480 & 0.0900 & 0.3590 & 0.0192 & 0.2898 & 0.0631 & 0.3290 \\
\cmidrule(lr){2-11}
& \multirow{3}{*}{500} & 0.15 & 0.0091 & 0.1306 & 0.1573 & 0.4381 & 0.0302 & 0.3054 & 0.1207 & 0.4272 \\
& & 0.30 & 0.0068 & 0.1219 & 0.1128 & 0.3748 & 0.0206 & 0.2639& 0.0843 & 0.3603 \\
& & 0.45 & 0.0053 & 0.1062 & 0.0740 & 0.3075 & 0.0136 & 0.2219 & 0.0539 &0.2917 \\
\midrule[1pt]
\multirow{9}{*}{M.4} & \multirow{3}{*}{100} & 0.15 & 0.0129 & 0.2171 & 0.0279 & 0.2915 & 0.0158 & 0.2633 & 0.0240 & 0.2922 \\
& & 0.30 & 0.0108 & 0.2037 & 0.0207 &0.2546 & 0.0125 & 0.2335 & 0.0172 & 0.2481 \\
& & 0.45 & 0.0079 & 0.1777 & 0.0142 & 0.2164 & 0.0092 & 0.1998 & 0.0113 & 0.2029 \\
\cmidrule(lr){2-11}
& \multirow{3}{*}{200} & 0.15 & 0.0072 & 0.1644 & 0.0199 & 0.2345 & 0.0096 & 0.2073 & 0.0205 & 0.2643 \\
& & 0.30 & 0.0060 & 0.1541 & 0.0149 & 0.2073 & 0.0077 & 0.1855 & 0.0146 & 0.2246 \\
& & 0.45 & 0.0044 & 0.1326 & 0.0104 & 0.1762 & 0.0057 & 0.1594 & 0.0095 & 0.1828 \\
\cmidrule(lr){2-11}
& \multirow{3}{*}{500} & 0.15 & 0.0032 & 0.1161 & 0.0112 & 0.1709 & 0.0047 & 0.1461 & 0.0166 & 0.2302 \\
& & 0.30 & 0.0026 & 0.1045 & 0.0087 & 0.1522 & 0.0037 & 0.1306 & 0.0119 & 0.1956 \\
& & 0.45 & 0.0019 & 0.0901 & 0.0063 & 0.1310 & 0.0027 & 0.1121 & 0.0078 & 0.1595 \\
\bottomrule[1.5pt]
\end{tabular}}
\end{table}

\end{appendices}

\section*{Author contributions}
\addcontentsline{toc}{section}{Author contributions}

CRediT: \textbf{G.\ Lyu}: Conceptualization, Methodology, Software, Writing -- original draft; \textbf{F.\ Ouimet}: Validation, Investigation, Writing -- review \& editing; \textbf{C.\ Feng}: Validation, Investigation, Writing -- review \& editing.

%\section*{Disclosure statement}
%\addcontentsline{toc}{section}{Disclosure statement}

%The authors declare no conflicts of interest.

\section*{Reproducibility}
\addcontentsline{toc}{section}{Reproducibility}

The \textsf{R} code implementing the proposed estimator and illustrating the simulation workflow is publicly available in the GitHub repository at~\url{https://github.com/GuanjieLyu/Tweedie-KDE}.

\section*{Funding}
\addcontentsline{toc}{section}{Funding}

%Funding in support of this work was provided by the Natural Sciences and Engineering Research Council of Canada through grants RGPIN-2026-04471 and DGECR-2026-00449 awarded to Fr\'ed\'eric Ouimet.

Funding for this work was provided in part by the Natural Sciences and Engineering Research Council of Canada (NSERC) through Discovery Grant RGPIN-2026-04471 and Discovery Launch Supplement DGECR-2026-00449 awarded to Fr\'ed\'eric Ouimet. Additional support for this work was provided by Research Nova Scotia through a New Health Investigator Grant and by NSERC through a Discovery Grant (RGPIN-2019-07212), both awarded to Cindy Feng. Guanjie Lyu and Cindy Feng also acknowledge support from the Mitacs Accelerate program.

\addcontentsline{toc}{section}{References}
\setlength{\bibsep}{0pt plus 0ex} % This line reduces the gap between references
\bibliographystyle{plainnat}
\bibliography{bib_clean}

\end{document}